\documentclass[12pt]{article}
\usepackage{amsmath}
\usepackage{graphicx,psfrag,epsf}
\usepackage{enumerate}
\usepackage{psfrag,epsf}
\usepackage{url} 
\usepackage{amsmath,amssymb,amsfonts,amsthm,mathtools,mathrsfs}

\usepackage{bm,bbm}
\usepackage{color}
\usepackage{subfig}
\usepackage{algorithm,algpseudocode}
\usepackage[symbol]{footmisc}
\usepackage{scalefnt}
\usepackage{authblk} 
\usepackage{multirow,centernot}
\usepackage[colorlinks=true, citecolor=blue, urlcolor=blue]{hyperref}
\usepackage{natbib}
\usepackage{dsfont}
\usepackage[title]{appendix}
\usepackage{newtxtext,newtxmath} 
\usepackage{mhchem} 

\input cyracc.def

\usepackage[left=1in,top=1.1in,right=0.5in,bottom=1in]{geometry}

\theoremstyle{definition}

\newtheorem{theorem}{Theorem}[section]

\newtheorem{proposition}[theorem]{Proposition}
\newtheorem{remark}[theorem]{Remark}

\makeatletter
\def\@seccntformat#1{\@ifundefined{#1@cntformat}%
	{\csname the#1\endcsname\quad}
	{\csname #1@cntformat\endcsname}
}
\makeatother

\newif\ifShowComments
\ShowCommentstrue
\def\strutdepth{\dp\strutbox}
\def\druk#1{\strut\vadjust{\kern-\strutdepth
        {\vtop to \strutdepth{%
                \baselineskip\strutdepth\vss
                        \llap{\hbox{#1}\quad}\null}}}}

\title{\bf
Modeling double bounded data based on correlated gamma random variables
}

\author{
\text{Roberto Vila}$^{1}$\thanks{Corresponding author: Roberto Vila, email: {rovig161@gmail.com} 
\newline
}, \,
\text{Felipe Quintino}$^{1}$
and
\text{Marcelo Bourguignon}$^{2}$ 
\\
{\small $^{1}$ Department of Statistics, University of Brasilia, Brasilia, Brazil}
\\
{\small $^{2}$ Department of Statistics, Federal University of Rio Grande do Norte, Natal, RN, Brazil}\\
}
\begin{document}
	\maketitle 	
	\begin{abstract}

Many types of bounded data defined on the unit interval arise naturally as ratios of the form 
$X/(X + Y)$. In the existing literature, the main statistical models proposed for this type of
bounded data typically based on the assumption that the random variables 
$X$ and $Y$ are independent. However, this assumption is often unrealistic in practical applications, where 
$X$ and $Y$ tend to be correlated due to shared underlying mechanisms or common sources of variability. 
In this paper, we overcome such limitations
and propose a model in which the marginal distributions of the two components are linked by
a copula, leading to a more flexible and realistic representation of unit-interval data.
In particular, in the proposed model, $X$ and $Y$ are 
{\color{black}
dependent gamma random variables whose joint distribution is specified via Morgenstern's bivariate distribution}, allowing for positive and negative correlations between the
components.
The mathematical properties and practical applications are rigorously investigated. 
{\color{black} The resulting distribution exhibits a wide range of shapes, accommodating different degrees of skewness and, for some parameter configurations, more complex density structures.}
A Monte Carlo simulation study is carried out that shows the good performance of the maximum likelihood estimator in several scenarios of parameter choices.
{The potential and limitations of efficient likelihood-based computations are also discussed.
We evaluate the effectiveness of the new model and its estimates in modeling real-world datasets related to economics.}

	\end{abstract}
	

	\section{Introduction}
	\noindent
	
{


Several natural indicators are measured as indicators, percentages,
proportions, ratios, and rates that are bounded on a certain interval, usually in the
bounded interval (0, 1).
The need for modeling and analyzing bounded data occurs in
many fields of real life such as  the Human Development Index, COVID-19 recovery rates, and percentages, among others.
Certainly, the
two parameter beta distribution is the most common distribution used in the literature to
describe data in the unit interval, especially because of its flexibility.

\cite{malik:67}
and \cite{ahuja:69}  both showed that if $X$ and $Y$ are independent random variables following
gamma distributions with shape parameters $\alpha > 0$ and
$\beta > 0$, and the same scale parameter $\sigma$, i.e., $X \sim \text{Gamma}(\alpha, \sigma)$ and $Y \sim \text{Gamma}(\beta, \sigma)$, then
\begin{equation}\label{eq:ratio_model}
Z \stackrel{d}{=} \frac{{\color{black} X}}{X + Y} \sim \text{Beta}(\alpha, \beta),
\end{equation}
namely, $Z$ is beta distributed with shape parameters $\alpha,\beta > 0$ and $\stackrel{d}{=}$ stands
for equality in distribution.
It should be noted that stochastic representations are important since they may justify some models
arising naturally in real situations. Furthermore, models generated by \eqref{eq:ratio_model} can produce distributions with different 
shapes, such as symmetry, bimodality, and others. This allows estimation procedures 
to better capture the inherent structures in the data.
However, the independence
assumption seems to be not suitable in real situations, where we would expect a relationship
between these variables.
The case of positively correlated gamma distributed components
has earlier been studied by \cite{Nadarajah2006} based on the Kibble bivariate
gamma distribution for $(X, Y)^\top$ \citep{kibble1941}. However, negatively correlated measurements are encountered
in numerous applications, for example in COVID-19 recovery rates, where 
$X$ and $Y$ are two random variables representing the number of confirmed COVID-19-related
deaths and COVID-19 cases without death result, and the independence
assumption seems to be not suitable in this case.

In this context, 
considering distributed components negatively and positively correlated, and
satisfying the stochastic representation in (\ref{eq:ratio_model}),
\cite{vila2024} proposed a class of unit-log-symmetric models to analyzing an internet access
data. \cite{vila2025} introduced a new distribution in the unit interval where $X$ and $Y$ are two correlated Birnbaum-Saunders
random variables. In the same way, \cite{vila2025novel} considered a novel unit‑asymmetric
distribution, where $X$ and $Y$ are two correlated Fréchet random
variables.

In this work, following the same stochastic formulation as the beta distribution, we introduce a new class of unit models that arises from a ratio \eqref{eq:ratio_model} of two correlated gamma random variables. We consider $(X,Y)^\top$ distributed as Morgenstern’s {\color{black} (which is also known in the literature as the Farlie-Gumbel-Morgenstern (FGM))} bivariate distribution with gamma margins \cite[cf.][]{Lai1978}. 
The new model extends the well-known Beta distribution with the addition of one parameter.
However, this extra parameter makes the model able to exhibit U- or J-shapes, symmetry, {\color{black} as well as more complex behaviors for certain parameter configurations.}
Moreover, information criteria on real-world dataset applications show that better fits are obtained with the new model, even with this extra parameter.
Maximum likelihood-based estimation procedures are discussed for the new distribution. A Monte Carlo simulation study is carried out, as well as applications to datasets well-consolidated in the literature, showing the versatility of the new model. {\color{black} It is well known that the Morgenstern copula is suitable for modeling bivariate data with weak dependence. However, it allows for analytically tractable derivations of the properties of the corresponding unit model. Moreover, our applications show that the proposed unit model can provide a good fit even when the original data exhibit strong correlation.}

Our results rely on special functions and their properties, which let us explicitly derive the probability density function (PDF), the cumulative distribution function (CDF), and key characteristics of the new probabilistic model.
Special functions are widely studied because they encompass many classes of functions and provide useful identities for integration, differentiation, and other operations. Some well-known examples are the Gauss hypergeometric ${}_2F_1$ function, Meijer's $G$-function, Fox’s $H$-function,  modified Bessel $K_\nu$-function of the 3rd kind  \cite[see Definitions 1.1, 1.5 and 1.6 and 1.11 in][]{MathaiSaxenaHaubold10}, and the double hypergeometric Appel functions $F_1, F_2, F_3$ and $F_4$ \citep{Appell1925}.

The remainder of this paper is organized as follows: 
Section \ref{sec:model} introduces the new model and plots of its shape for different parameter choices.
 Section \ref{sec:properties} deals with the
derivation of the new unit distribution and its mathematical properties, such as symmetry, stochastic representation, characterization as a ratio and moments.
In Section \ref{esti}, we propose the estimation of the
parameters.
In Section \ref{sec:simulation}, we discuss Monte Carlo simulations to evaluate the proposed estimator, its computational challenges, and alternative frameworks.
In Section \ref{sec:applications}, we address modeling real-world scenarios involving an income consumption dataset. The final section presents our conclusions, {\color{black}while Appendix A and B present, respectively, some additional results on the Appel $F_2$ function and the competing models.}
}

	\section{The extended beta distribution}\label{sec:model}
		\noindent

	We say that an absolutely continuous random variable $Z$, with support $(0, 1)$, follows an extended beta (EB) distribution with parameter vector $\boldsymbol{\theta} = (\alpha,\beta,\rho)^\top$, $\alpha> 0$ $\beta > 0$, and $-1\leqslant \rho\leqslant 1$, denoted by $Z\sim {\rm EB}(\boldsymbol{\theta})$, if its PDF 
    is given by (for $0 < z < 1$)
	\begin{align}\label{pdf-main}
		f_Z(z;\boldsymbol{\theta})
		&=
		(1+\rho) \, 
	{z^{\alpha-1} (1-z)^{\beta-1} \over 	
		\textrm{B}(\alpha,\beta)}
		\nonumber
	\\[0,2cm]
	&-
    \rho\,
{2\Gamma(2\alpha+\beta)\over \alpha\Gamma^2(\alpha)\Gamma(\beta)} 
\,
{z^{2\alpha-1}(1-z)^{\beta-1}\over(1+z)^{2\alpha+\beta}}
\,
_2F_1\left(1,2\alpha+\beta;1+\alpha;{z\over 1+z}\right)
			\nonumber
	\\[0,2cm]
	&-
    \rho\,
{2\Gamma(\alpha+2\beta)
	\over \beta \Gamma(\alpha)\Gamma^2(\beta)} 
\,
{z^{\alpha-1}(1-z)^{2\beta-1}\over (2-z)^{\alpha+2\beta}}\,
_2F_1\left(1,\alpha+2\beta;1+\beta;{1-z\over 2-z}\right)
			\nonumber
	\\[0,2cm]
	&+
    \rho\,
	{4^{1-\alpha-\beta} \Gamma(2\alpha+2\beta)\over \alpha \beta \Gamma^2(\alpha)\Gamma^2(\beta)}\,
 z^{2\alpha-1} (1-z)^{2\beta-1}
 	\,
 	F_2\left(2\alpha+2\beta,1,1;\alpha+1,\beta+1;{z\over 2}, {1-z\over 2}\right),
	\end{align}
%
	where 
	$\Gamma(\alpha)=\int_0^\infty t^{\alpha-1}\exp(-t){\rm d}t$ is the complete gamma function,
	$\mathrm {B} (\alpha _{1},\alpha_2)=\Gamma(\alpha _{1})\Gamma(\alpha _{2})/\Gamma(\alpha _{1}+\alpha _{2})$ is the complete beta function, 
	\begin{align*}
		{\displaystyle {}_{2}F_{1}(a,b;c;x)
			=
			\sum _{n=0}^{\infty }{\frac {(a)_{n}(b)_{n}}{(c)_{n}}}{\frac {x^{n}}{n!}}}
	\end{align*}
	is the hypergeometric function, 
	\begin{align}\label{eq:F2}
		F_2(a,b_1,b_2;c_1,c_2;x,y)
		=
		\sum_{m,n=0}^\infty
		{(a)_{m+n} (b_1)_m (b_2)_n\over (c_1)_m (c_2)_n}\, 
		{x^m\over m!} \, {y^n\over n!}
	\end{align} 
	 denotes the Appell F2 double hypergeometric function, and
	 $(q)_n=\Gamma(q+n)/\Gamma(q)$ is the Pochhammer symbol. {Conditions for the convergence of the series and several properties of these special functions are given in \cite{MathaiSaxenaHaubold10} and  \cite{Ananthanarayan2023}.}

	\begin{remark}
Note that when $\rho=0$ the EB model in \eqref{pdf-main} becomes the classic beta distribution.
Therefore, the EB model can be interpreted as a new extended beta distribution.
	\end{remark}

Figure \ref{fig:pdf} shows the PDF $f_Z(z{\color{black};\boldsymbol{\theta}})$ behavior for some parameter choices.
Note that we can have a symmetric distribution centered around 0.5 (if $\alpha=\beta$, see Figure \ref{fig:pdf_symm}). It is also possible to obtain asymmetric or bimodal (Figure \ref{fig:pdf_bimodal}) models, depending on the choice of parameters.
    
\begin{figure}[H]
	\centering
	\includegraphics[width=0.9\linewidth, height=7cm]{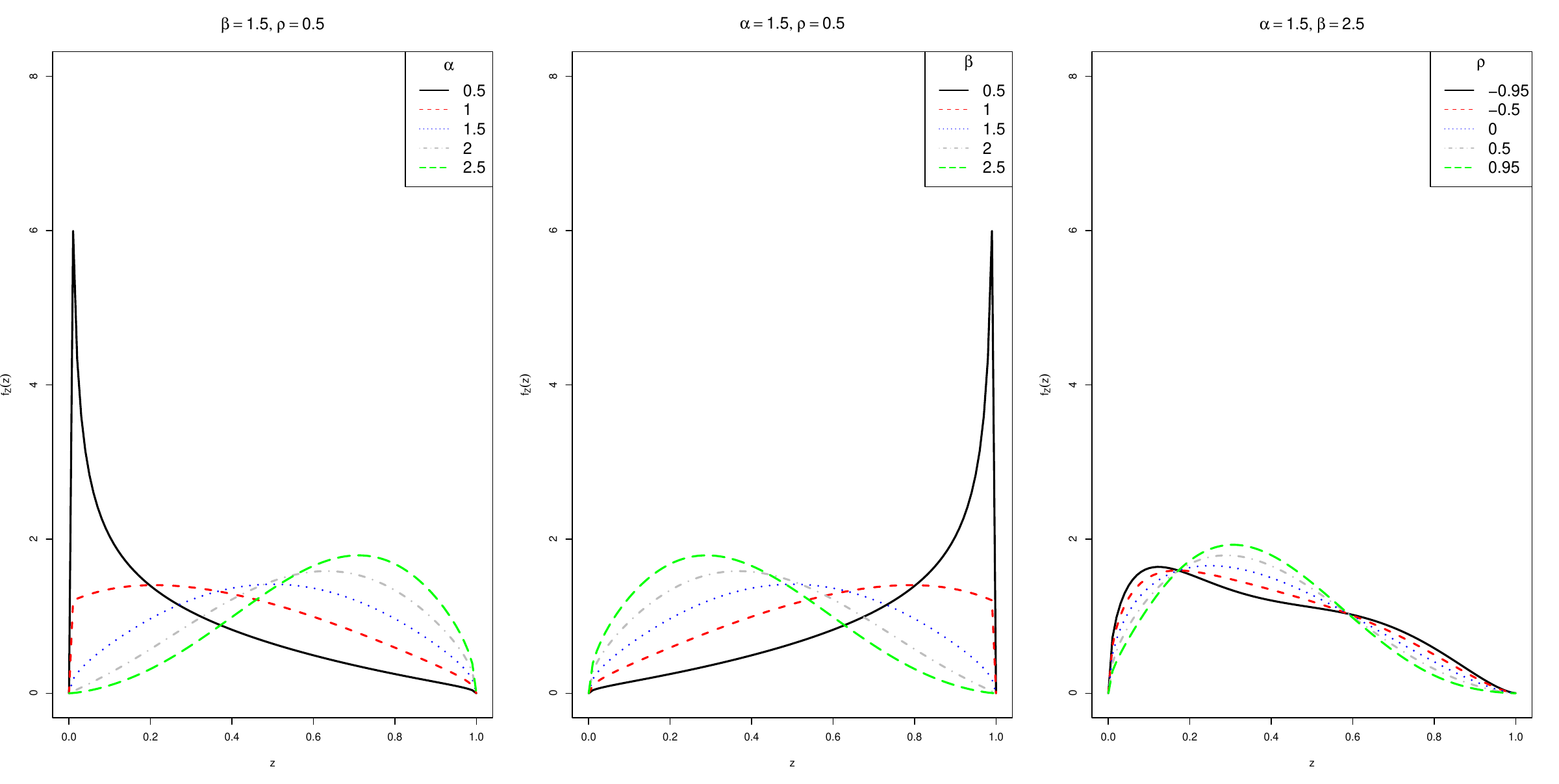}
	\caption{Plot of the PDF $f_Z$ with varying parameters $\alpha$ (left), $\beta$ (middle) and $\rho$ (right).}
	\label{fig:pdf}
\end{figure}
	
\begin{figure}[H]
	\centering
    \includegraphics[width=0.9\linewidth, height=9cm]{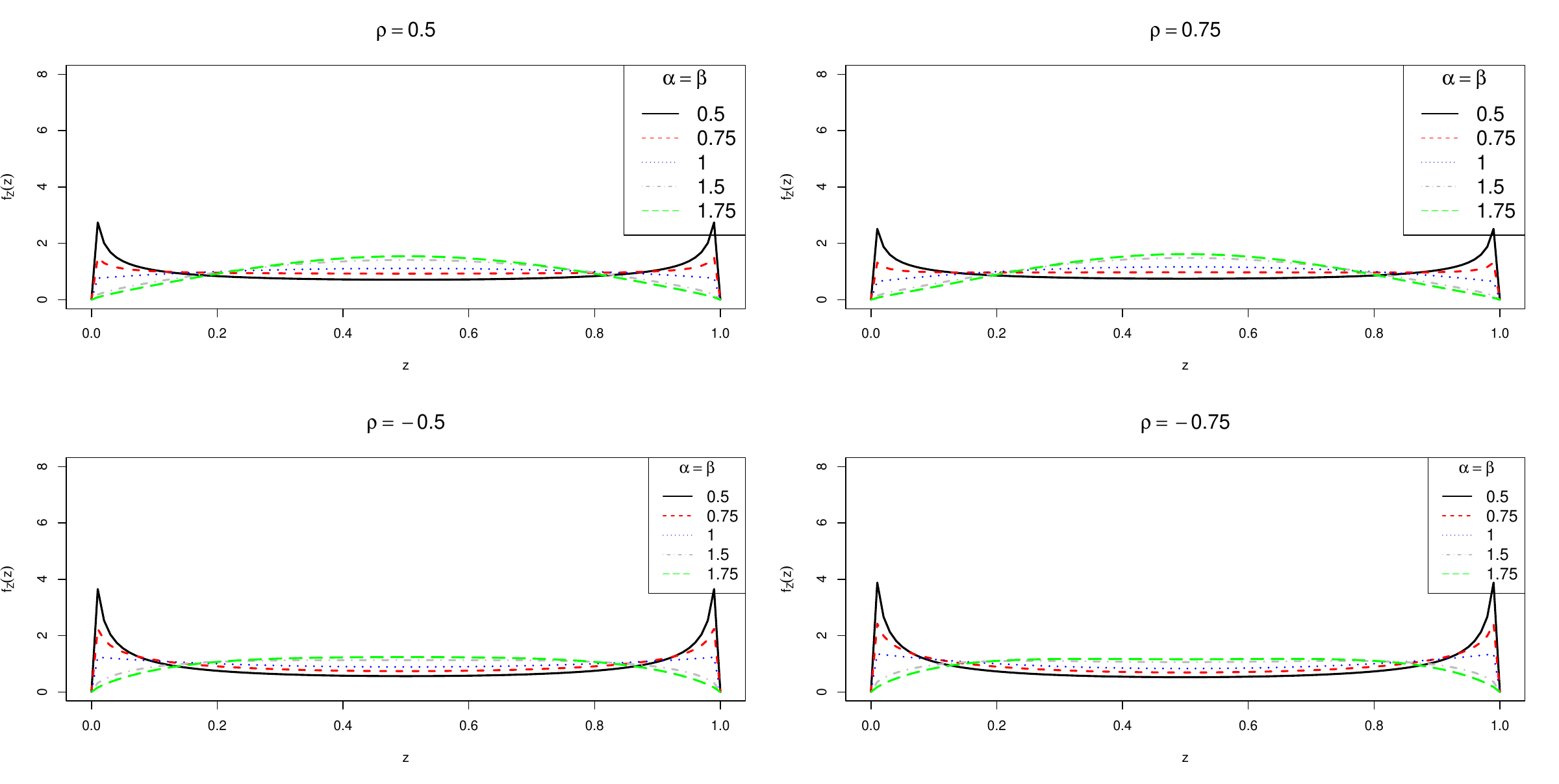}
	\caption{Plot of the PDF $f_Z$ with varying parameters $\alpha=\beta$ for $\rho\in\{-0.5, -0.75, 0.5, 0.75\}$.}
	\label{fig:pdf_symm}
\end{figure}

\begin{figure}[H]
	\centering
     \includegraphics[width=1.0\linewidth, height=10cm]{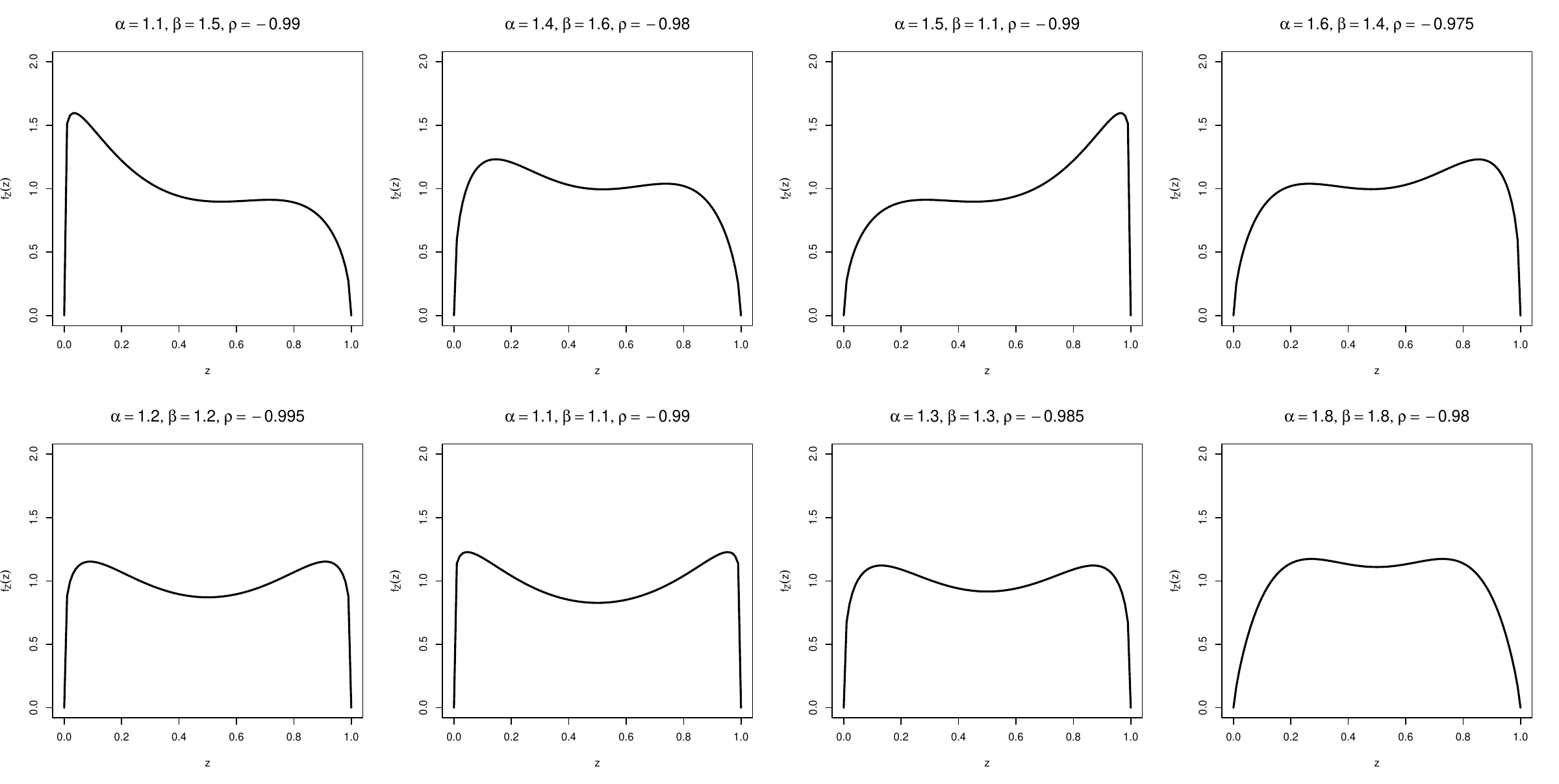}
	\caption{Plot of the PDF $f_Z$ with varying parameters $\alpha$, $\beta$ and $\rho$, illustrating bimodality.}
	\label{fig:pdf_bimodal}
\end{figure}

In Figures \ref{fig:pdf}--\ref{fig:pdf_bimodal} we note that:
\begin{itemize}
    \item for $\alpha >\beta$, the shapes are concentrated on the right side of the interval $[0,1]$;
    
    \item for $\alpha <\beta$ , the concentration is in the left side;
    
    \item when $\rho \approx -1$, bimodality can be observed. However, the two modes are not completely separated. After the first mode, the frequency does not drop near zero before rising again toward the second mode.

\end{itemize}

\begin{remark}
  Usual techniques for establishing unimodality or bimodality rely on analyzing the sign of the PDF’s derivative and its roots \cite[see, e.g., Theorem 2.1 in ][]{vila2024model}. However, the PDF \eqref{pdf-main} involves the special functions ${}_2F_1$ and $F_2$, which can make root analysis either impossible or extremely difficult from an analytical approach.
\end{remark}

{\color{black}

As closed-form conditions are unattainable, a more systematic numerical exploration of the bimodality region is introduced in Algorithm \ref{algo:bimodality}.
The algorithm can be run over a pre-specified parameter grid $\Theta=\{\boldsymbol{\theta}^{(1)},\cdots,\boldsymbol{\theta}^{(k)}\}$. Furthermore, an additional step can be incorporated to evaluate the depth of the valley between the two detected peaks.
The parameters shown in Figure \ref{fig:pdf_bimodal} were selected after several rounds of testing and refinement of Algorithm \ref{algo:bimodality}.

\begin{algorithm}[htb!]
\caption{Searching for bimodality}
\label{algo:bimodality}
\begin{algorithmic}[1]
    \Statex \textbf{Input:} a vector $\boldsymbol{\theta}\in(0,\infty)\times(0,\infty)\times[-1,1]$  ;
    \Statex \textbf{Output:} points of possible bimodality (if they exist).

 \State Fix the values $z_1, z_2, \cdots, z_n \in[0,1]$;

\State Define the lists
$$ P \gets \texttt{NULL};  V \gets \texttt{NULL};$$
\State Compute $\Delta_i f = f(z_{i+1};\boldsymbol{\theta})- f(z_i;\boldsymbol{\theta})$, for $i=1,\cdots, n-1$;
\For{$i=2,\cdots,n-1$}
   \If{$\Delta_{i-1} > 0 ~\&~ \Delta_{i} < 0$}
            \State $P \gets z_i$;
        \EndIf
   \If{$\Delta_{i-1} < 0 ~\&~ \Delta_{i} > 0$}
            \State $V \gets z_i$;
        \EndIf        
    \EndFor
\If{$\#P\geq 2$}
\State Compute
$$p_{1} = \arg\max_{z_i\in P} f(z_i; \boldsymbol{\theta}),~~ p_{2} = \arg\max_{z_i\in P;~~  z_i\neq p_{1}} f(z_i; \boldsymbol{\theta})~~\text{and}~~v_* = \arg\min_{p_2<z_i<p_1}f(z_i; \boldsymbol{\theta});$$
\State \textbf{Return} $p_1$, $p_2$, and $v_*$.
\Else
    \State \textbf{Return} ``unimodal''.
\EndIf
\end{algorithmic}
\end{algorithm}
}

	\section{Some basic properties}\label{sec:properties}
    \noindent
    
	In this section we establish some basic mathematical properties of the extended beta distribution.

	\subsection{Symmetry}\label{Symmetry}
	\noindent
    
Let $Z\sim {\rm EB}(\boldsymbol{\theta})$.
A simple observation shows that  (for $0<w<1/2$)
	\begin{align}\label{pdf-main-1}
	&f_Z(z;\boldsymbol{\theta})\big\vert_{z={1\over 2}-w}
	=
	(1+\rho) \, 
	{({1\over 2}-w)^{\alpha-1} ({1\over 2}+w)^{\beta-1} \over 	
		\textrm{B}(\alpha,\beta)}
	\nonumber
	\\[0,2cm]
&-
\rho\,
{2\Gamma(2\alpha+\beta)\over \alpha\Gamma^2(\alpha)\Gamma(\beta)} 
\,
{({1\over 2}-w)^{2\alpha-1}({1\over 2}+w)^{\beta-1}\over({3\over 2}-w)^{2\alpha+\beta}}
\,
_2F_1\left(1,2\alpha+\beta;1+\alpha;{1-2w\over 3-2w}\right)
			\nonumber
	\\[0,2cm]
	&-
    \rho\,
{2\Gamma(\alpha+2\beta)
	\over \beta \Gamma(\alpha)\Gamma^2(\beta)} 
\,
{({1\over 2}-w)^{\alpha-1}({1\over 2}+w)^{2\beta-1}\over ({3\over 2}+w)^{\alpha+2\beta}}\,
_2F_1\left(1,\alpha+2\beta;1+\beta;{1+2w\over 3+2w}\right)
			\nonumber
	\\[0,2cm]
	&+
    \rho\,
	{4^{1-\alpha-\beta} \Gamma(2\alpha+2\beta)\over \alpha \beta \Gamma^2(\alpha)\Gamma^2(\beta)}\,
 {
 \left({1\over 2}-w\right)^{2\alpha-1} \left({1\over 2}+w\right)^{2\beta-1}
 	\,
 	F_2\left(2\alpha+2\beta,1,1;\alpha+1,\beta+1;{1-2w\over 4}, {1+2w\over 4}\right)
    }
\end{align}
and
	\begin{align}\label{pdf-main-2}
	&f_Z(z;\boldsymbol{\theta})\big\vert_{z={1\over 2}+w}
	=
	(1+\rho) \, 
	{({1\over 2}+w)^{\alpha-1} ({1\over 2}-w)^{\beta-1} \over 	
		\textrm{B}(\alpha,\beta)}
	\nonumber
	\\[0,2cm]
	&-
    \rho\,
{2\Gamma(2\alpha+\beta)\over \alpha\Gamma^2(\alpha)\Gamma(\beta)} 
\,
{({1\over 2}+w)^{2\alpha-1}({1\over 2}-w)^{\beta-1}\over({3\over 2}+w)^{2\alpha+\beta}}
\,
_2F_1\left(1,2\alpha+\beta;1+\alpha;{1+2w\over 3+2w}\right)
			\nonumber
	\\[0,2cm]
	&-
    \rho\,
{2\Gamma(\alpha+2\beta)
	\over \beta \Gamma(\alpha)\Gamma^2(\beta)} 
\,
{({1\over 2}+w)^{\alpha-1}({1\over 2}-w)^{2\beta-1}\over ({3\over 2}-w)^{\alpha+2\beta}}\,
_2F_1\left(1,\alpha+2\beta;1+\beta;{1-2w\over 3-2w}\right)
			\nonumber
	\\[0,2cm]
	&+
    \rho\,
	{4^{1-\alpha-\beta} \Gamma(2\alpha+2\beta)\over \alpha \beta \Gamma^2(\alpha)\Gamma^2(\beta)}\,
 {
 \left({1\over 2}+w\right)^{2\alpha-1} \left({1\over 2}-w\right)^{2\beta-1}
 	\,
 	F_2\left(2\alpha+2\beta,1,1;\alpha+1,\beta+1;{1+2w\over 4}, {1-2w\over 4}\right)
    }.
\end{align}
As the the Appell F2 double hypergeometric function is symmetric, that is, $F_2(a,b_1,b_2;c_1,c_2;x,y)=F_2(a,b_1,b_2;c_1,c_2;y,x)$, by taking $\alpha=\beta$ in the identities \eqref{pdf-main-1} and \eqref{pdf-main-2}, we have (for $0<w<1/2$)
\begin{align*}
	f_Z(z;\boldsymbol{\theta})\big\vert_{z={1\over 2}-w}
	=
	f_Z(z;\boldsymbol{\theta})\big\vert_{z={1\over 2}+w}.
\end{align*}
That is, the EB PDF in \eqref{pdf-main} is symmetric around
$z_0 = 1/2$ provided $\alpha=\beta$ (see Figure \ref{fig:pdf_symm}). Moreover, in this case, the median and the mean
of a EB distribution both occur at $z_0 = 1/2$.

	\subsection{EB model arising as a ratio of correlated random variables}
	\label{EB model arising as a ratio}
	\noindent
    
In this subsection, we establish a fundamental property of the EB model, in which the random variable $Z$ in \eqref{pdf-main} is represented as a ratio  of the form 
 \begin{align}\label{rep-stoch-1-1}
	Z\stackrel{d}{=}{X\over X+Y},
\end{align}
where $\stackrel{d}{=}$ means equality in distribution and the random vector $(X, Y)^\top$ follows a Morgenstern's bivariate distribution (with gamma margins) \citep{Lai1978}.
 
{\color{black}
Define the bivariate function $C:[0,1]\times[0,1] \rightarrow [0,1]$ as
\begin{equation}\label{eq:copula_Morgenstern}
    C(u_1,u_2) = u_1u_2\left[ 1+ \rho(1-u_1)(1-u_2)\right], ~~(u_1,u_2)\in[0,1]\times[0,1] ~\mbox{and}~-1\leqslant \rho\leqslant 1. 
\end{equation}
Although the copula $C(\cdot,\cdot)$ defined above has limited dependence flexibility, it was adopted mainly due to its analytical tractability. In order to provide richer dependence structures, we have included in the real-data analysis (Section \ref{sec:applications}) some additional flexible extensions of the copula $C(\cdot,\cdot)$.
}
 
Following \cite{Lai1978}, a random vector $(X, Y)^\top$ has a  Morgenstern's bivariate  distribution if its cumulative distribution function (CDF) is of the form
\begin{align}\label{Morgenstern-cdf}
	F_{X,Y}(x,y)
    {\color{black}
    \, =
    C(F_X(x),F_Y(y))
    }
	=
	F_X(x)F_Y(y)\big[1+\rho\{1-F_X(x)\}\{1-F_Y(y)\}\big],
	\quad -1\leqslant\rho\leqslant 1,
\end{align}
having $F_X(x)$ and $F_Y(y)$ as marginal distribution functions. It is clear that the corresponding joint PDF of $(X, Y)^\top$ is given by 
 \begin{align}\label{Morgenstern}
 	f_{X,Y}(x,y)
 	=
 	f_X(x)f_Y(y)
 	\big[1+\rho\{2F_X(x)-1\}\{2F_Y(y)-1\}\big].
 \end{align}
 Note that when $\rho=0$, $X$ and $Y$ are independent. 

In what follows, we verify that the random variable $Z$ in \eqref{pdf-main} satisfies the stochastic representation \eqref{rep-stoch-1-1}, where $(X, Y)^\top$ has Morgenstern's bivariate distribution \eqref{Morgenstern-cdf}-\eqref{Morgenstern}, and $X$ and $Y$ have gamma distributions with the same rate, more precisely, $X~\sim{\rm Gamma}(\alpha,\theta)$ and $Y~\sim{\rm Gamma}(\beta,\theta)$, $\alpha>0,\beta>0,\theta>0$.

By fact, using the Jacobian method, we have (for $s=1/z-1$ and  $0<z<1$)
 \begin{align*}
 	f_{{X\over X+Y}}(z)
 	= 
 	(s+1)^2
 	\int_0^\infty
 	x f_{X,Y}(x,sx)
 	{\rm d}x.
 \end{align*}
 Using \eqref{Morgenstern} in the above identity, we have the following
 \begin{align}\label{iden-pdf-ratio}
 	f_{{X\over X+Y}}(z)
 	&=
 	(1+\rho)
 	(s+1)^2
 	\int_0^\infty
 	x
 	f_X(x) f_Y(sx)
 	{\rm d}x
 	-
 	2
 	\rho
 	(s+1)^2
 	\int_0^\infty
 	[xf_X(x) f_Y(sx)]
 	F_X(x)
 	{\rm d}x
 	\nonumber
 	\\[0,2cm]
 	&-
 	2
 	\rho
 	(s+1)^2
 	\int_0^\infty
 	[xf_X(x) f_Y(sx)]
 	F_Y(sx)
 	{\rm d}x
 	+
 	4
 	\rho
 	(s+1)^2
 	\int_0^\infty
 	[xf_X(x) f_Y(sx)] 
 	 F_X(x) F_Y(sx)
 	{\rm d}x
 	\nonumber
 	\\[0,2cm]
 	&\equiv
 	(1+\rho)
 	K_1(z)
 	 	-
 	\rho
 	K_2(z)
 	-
 	\rho
 	K_3(z)
 	+
 	\rho
 	K_4(z),
 \end{align}
 where
 \begin{align}\label{def-K}
 			\begin{array}{llll}
&K_1(z)\equiv  	\displaystyle
(s+1)^2
\int_0^\infty
x
f_X(x) f_Y(sx)
{\rm d}x,
\\[0,5cm]
&K_2(z)\equiv \displaystyle
2
(s+1)^2
\int_0^\infty
[xf_X(x) f_Y(sx)]
F_X(x)
{\rm d}x,
\\[0,5cm]
&K_3(z)\equiv \displaystyle
2
(s+1)^2
\int_0^\infty
[xf_X(x) f_Y(sx)]
F_Y(sx)
{\rm d}x,
\\[0,5cm]
&K_4(z)\equiv  	\displaystyle
4
(s+1)^2
\int_0^\infty
[xf_X(x) f_Y(sx)] 
F_X(x) F_Y(sx)
{\rm d}x.
		\end{array}
 \end{align}
 
 	Since $X~\sim{\rm Gamma}(\alpha,\theta)$, $Y~\sim{\rm Gamma}(\beta,\theta)$ and
 \begin{align*}
 	x
 	f_X(x;\alpha,\theta) f_Y(sx;\beta,\theta)
 	=
 	{\theta^{\alpha+\beta} s^{\beta-1}\over \Gamma(\alpha)\Gamma(\beta)} \, 
 	x^{\alpha+\beta-1} \exp[-(1+s)\theta x],
 \end{align*}
 we have
 \begin{align*}
 &K_1(z)
 =  	
 (s+1)^2 \,
 	{\theta^{\alpha+\beta} s^{\beta-1}\over \Gamma(\alpha)\Gamma(\beta)} 
 \int_0^\infty
x^{\alpha+\beta-1} \exp[-(1+s)\theta x]
 {\rm d}x,
 \\[0,2cm]
 &K_2(z)=
 2
 (s+1)^2 \,
 {\theta^{\alpha+\beta} s^{\beta-1}\over \Gamma^2(\alpha)\Gamma(\beta)} 
 \int_0^\infty
x^{\alpha+\beta-1} \exp[-(1+s)\theta x]
\gamma(\alpha,\theta x)
 {\rm d}x,
 \\[0,2cm]
 &K_3(z)=
 2
 (s+1)^2 \,
 {\theta^{\alpha+\beta} s^{\beta-1}\over \Gamma(\alpha)\Gamma^2(\beta)} 
 \int_0^\infty 
x^{\alpha+\beta-1} \exp[-(1+s)\theta x]
 \gamma(\beta,\theta s x)
 {\rm d}x,
 \\[0,2cm]
 &K_4(z)= 	
 4
 (s+1)^2 \,
 {\theta^{\alpha+\beta} s^{\beta-1}\over \Gamma^2(\alpha)\Gamma^2(\beta)} 
 \int_0^\infty 
x^{\alpha+\beta-1} \exp[-(1+s)\theta x]
 \gamma(\alpha,\theta x)\gamma(\beta,\theta sx)
 {\rm d}x.
 \end{align*}
 
 It is clear that
 \begin{align}\label{K1-def}
 	K_1(z)
 	=
 	{z^{\alpha-1} (1-z)^{\beta-1}\over B(\alpha,\beta)}.
 \end{align}
 
 By using the identity of
Proposition \ref{prop-app-3}:
\begin{align}\label{prop-app-3-1}
 	\int_0^\infty 
    x^{a-1}
    \exp(-sx)\gamma(b,\theta x){\rm d}x
    =
    {\theta^b\Gamma(a+b)\over b(s+\theta)^{a+b}}
\,_{2}F_{1}\left(a+b,1;b+1;{\theta\over s+\theta}\right),
\end{align}
 we have
 \begin{align}\label{K2-def}
K_2(z)
%
=
{2\Gamma(2\alpha+\beta)\over \alpha\Gamma^2(\alpha)\Gamma(\beta)} 
\,
{z^{2\alpha-1}(1-z)^{\beta-1}\over(1+z)^{2\alpha+\beta}}
\,
_2F_1\left(1,2\alpha+\beta;1+\alpha;{z\over 1+z}\right)
 \end{align}
 and
 \begin{align}\label{K3-def}
 K_3(z)
=
{2\Gamma(\alpha+2\beta)
	\over \beta \Gamma(\alpha)\Gamma^2(\beta)} 
\,
{z^{\alpha-1}(1-z)^{2\beta-1}\over (2-z)^{\alpha+2\beta}}\,
_2F_1\left(1,\alpha+2\beta;1+\beta;{1-z\over 2-z}\right).
 \end{align}
 
 Furthermore, by using the identity of Proposition \ref{prop-app-2}:
\begin{multline}\label{Prudnikov2002-2}
 	\int_0^\infty 
    x^{a-1}
    \exp(-sx)\gamma(b,\theta x)\gamma(c,\xi x){\rm d}x
    \\[0,2cm]
 	=
 	{\theta^b \xi^c \Gamma(a+b+c)\over bc
    (s+\theta+\xi)^{a+b+c}} \,
 	F_2\left(a+b+c,1,1;b+1,c+1;{\theta\over s+\theta+\xi},{\xi\over s+\theta+\xi}\right),
\end{multline}
 we obtain 
 \begin{align}\label{K4-def}
 K_4(z)
 =
{4^{1-\alpha-\beta} \Gamma(2\alpha+2\beta)\over \alpha\beta \Gamma^2(\alpha)\Gamma^2(\beta)}\,
 {
 z^{2\alpha-1} (1-z)^{2\beta-1}
 	\,
 	F_2\left(2\alpha+2\beta,1,1;\alpha+1,\beta+1;{z\over 2}, {1-z\over 2}\right)
 }.
 \end{align}
 
 By replacing \eqref{K1-def}, \eqref{K2-def}, \eqref{K3-def} and \eqref{K4-def} in \eqref{iden-pdf-ratio}, we get (for $0<z<1$)
 \begin{align*}
 	f_{{X\over X+Y}}(z) 	
 	&=
 	(1+\rho)\,
 	{z^{\alpha-1} (1-z)^{\beta-1}\over B(\alpha,\beta)}
 	\\[0,2cm]
 	&-
    \rho\,
{2\Gamma(2\alpha+\beta)\over \alpha\Gamma^2(\alpha)\Gamma(\beta)} 
\,
{z^{2\alpha-1}(1-z)^{\beta-1}\over(1+z)^{2\alpha+\beta}}
\,
_2F_1\left(1,2\alpha+\beta;1+\alpha;{z\over 1+z}\right)
 	 	\\[0,2cm]
 	&-
    \rho\,
{2\Gamma(\alpha+2\beta)
	\over \beta \Gamma(\alpha)\Gamma^2(\beta)} 
\,
{z^{\alpha-1}(1-z)^{2\beta-1}\over (2-z)^{\alpha+2\beta}}\,
_2F_1\left(1,\alpha+2\beta;1+\beta;{1-z\over 2-z}\right)
 	 	\\[0,2cm]
 	&+
    \rho\,
{4^{1-\alpha-\beta} \Gamma(2\alpha+2\beta)\over \alpha \beta \Gamma^2(\alpha)\Gamma^2(\beta)}\,
 {
 z^{2\alpha-1} (1-z)^{2\beta-1}
 	\,
 	F_2\left(2\alpha+2\beta,1,1;\alpha+1,\beta+1;{z\over 2}, {1-z\over 2}\right)
    }
    \\[0,2cm]
 	&=	f_Z(z;\boldsymbol{\theta}).
 \end{align*}
 
 Therefore, we have proven that
 \begin{align}\label{eq-pdfs}
 	 	f_{{X\over X+Y}}(z)=f_Z(z;\boldsymbol{\theta}), \quad 0<z<1.
 \end{align}
In other words, the stochastic representation for $Z$ given in \eqref{rep-stoch-1-1} is valid.

{Figure \ref{fig:jointPDF} shows the joint PDF of Morgenstern's bivariate distribution with Gamma margins and contour lines of the joint density, for fixed $\alpha=1.1$ and
$\beta=1.5$ and varying $\rho\in\{-0.99, 0, 0.75\}$.
Note that when $\rho =-0.99$, a possible local maximum appears near the global maximum. This corresponds to the same parameter configuration that produced bimodality in one of the plots in Figure \ref{fig:pdf_bimodal}. } 

\begin{figure}[!htbp]
    \centering
    \subfloat[$\rho=-0.99$]{\includegraphics[width=.5\linewidth]{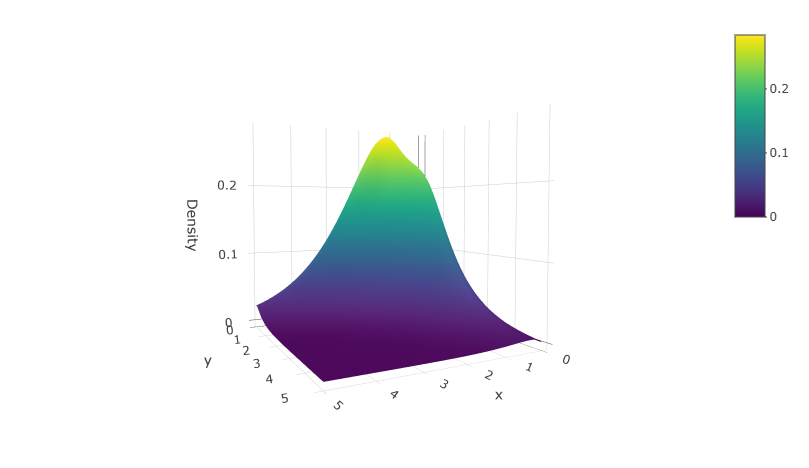}}
    \subfloat[$\rho=-0.99$]{\includegraphics[width=.5\linewidth]{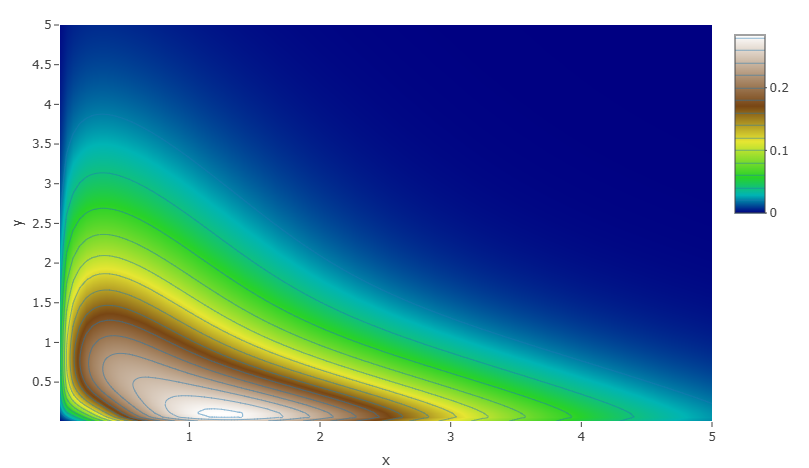}} \\
    \subfloat[$\rho=0$]{\includegraphics[width=.5\linewidth]{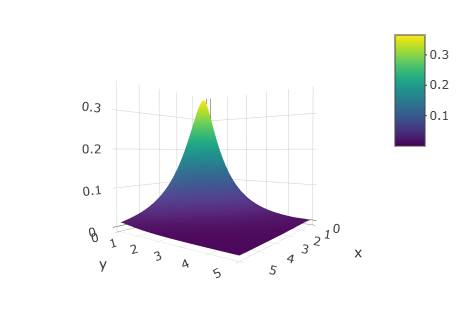}}
    \subfloat[$\rho=0$]{\includegraphics[width=.5\linewidth]{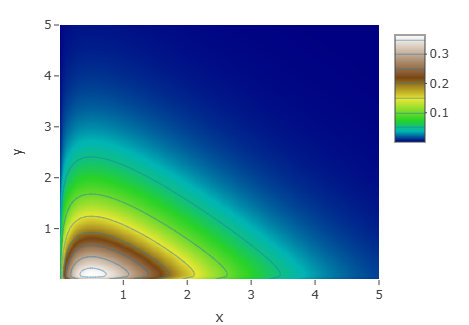}}\\
        \subfloat[$\rho=0.75$]{\includegraphics[width=.5\linewidth]{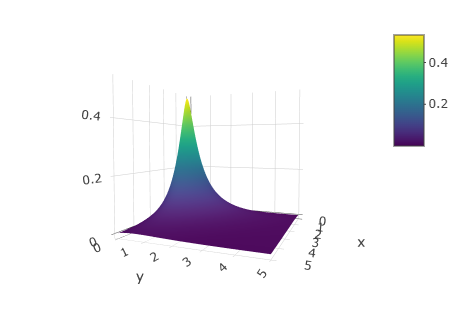}}
    \subfloat[$\rho=0.75$]{\includegraphics[width=.5\linewidth]{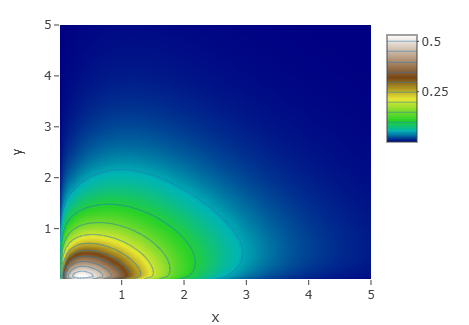}} 
    \caption{Joint PDF of Morgenstern's bivariate distribution with Gamma margins and contour lines of the joint density for $\rho=-0.99$ in (a) and (b), $\rho=0$ in (c) and (d), and $\rho=0.75$ in (e) and (f). $\alpha=1.1$ and
$\beta=1.5$.}
    \label{fig:jointPDF}
\end{figure}

\subsection{Other characterizations of EB model}
\noindent

Let $X$ and $Y$ be two (absolutely) continuous positive random variables with CDFs $F_X$ and $F_Y$, respectively. Based on $F_X$ and $F_Y$ and its respective PDFs $f_X$ and $f_Y$, we define two new densities $g_X$ and $g_Y$ as follows:
\begin{align} \label{def-g}
	g_X(x)=2f_X(x)F_X(x), 
	\quad
	g_Y(y)=2f_Y(y)F_Y(y),
	\quad x, y>0.
\end{align}

%

\begin{proposition}\label{remark-1}
	The non-negative functions $f_1, f_2,f_3$ and $f_4$ defined as (for $s=1/z-1$ and $0<z<1$)
	\begin{align}\label{def-f-1}
		\begin{array}{llll}
			&f_1(z)
			\equiv
			\displaystyle
			(s+1)^2
			\int_0^\infty
			x
			f_X(x) f_Y(sx)
			{\rm d}x,
			\\[0,5cm]
			&f_2(z)
			\equiv		\displaystyle
			(s+1)^2
			\int_0^\infty
			x
			g_X(x) f_Y(sx)
			{\rm d}x,
			\\[0,5cm]
			&f_3(z)
			\equiv	\displaystyle
			(s+1)^2
			\int_0^\infty
			x
			f_X(x) 
			g_Y(sx)
			{\rm d}x,
			\\[0,5cm]
			&f_4(z)
			\equiv		\displaystyle
			(s+1)^2
			\int_0^\infty
			x
			g_X(x) g_Y(sx)
			{\rm d}x,
		\end{array}
	\end{align}
	are PDFs.  In the above, $g_X$ and $g_Y$ are as in \eqref{def-g}.
\end{proposition}
\begin{proof}
	Let	$p$ and $q$ be two arbitrary densities with positive support. In what follows, we prove that the function $g$ defined as (for $s=1/z-1$ and $0<z<1$)
	\begin{align*}
		g(z)\equiv 
		(s+1)^2
		\int_0^\infty
		x p(x) q(sx)
		{\rm d}x	
	\end{align*}
	is a PDF.
	Indeed, since
	\begin{align*}
		\int_0^1 
		g(z) 
		{\rm d}z
		=
		\int_0^\infty
		x p(x)
		\left[
		\int_{0}^{1}
		{1\over z^2}\,
		q\left(\left({1\over z}-1\right)x\right)
		{\rm d}z
		\right]
		{\rm d}x,
	\end{align*}
	taking the change of variable $y=(1/z-1)x$, the right-hand  integral is written as	
	\begin{align*}
		&=
		\int_0^\infty
		x
		p(x) 
		\left[
		{1\over x}
		\int_{0}^{\infty}
		q(y)
		{\rm d}y
		\right]
		{\rm d}x
		=1,
	\end{align*}
	where in the last equality, we have used the fact that $p$ and $q$ are PDFs.
	
	Hence, taking (i) $p=f_X$ and $q=f_Y$, (ii) $p=g_X$ and $q=f_Y$, (iii) $p=f_X$ and $q=g_Y$, and (iv) $p=g_X$ and $q=g_Y$,  it follows that $f_1, f_2,f_3$ and $f_4$ in \eqref{def-f-1} are PDFs. 
\end{proof}

Since $f_i(z)=K_i(z)$, $i=1,2,3,4$, where $K_i$ is as in \eqref{def-K}, from \eqref{iden-pdf-ratio}, we have (for $0<z<1$)
		\begin{align}\label{id-pds}
		 	f_{{X\over X+Y}}(z)
		 	=
		 	 	(1+\rho)
		 	f_1(z)
		 	-
		 	\rho
		 	f_2(z)
		 	-
		 	\rho
		 	f_3(z)
		 	+
		 	\rho
		 	f_4(z).
		\end{align}
		
If furthermore,	 $(X, Y)^\top$ follows a Morgenstern's bivariate  distribution (with gamma margins), then,	by \eqref{eq-pdfs} and \eqref{id-pds}, we have
\begin{align*}
	f_Z(z;\boldsymbol{\theta})
	=
	f_{{X\over X+Y}}(z)
	=
	\sum_{i=1}^{4} \pi_if_i(z),  
	\quad 0<z<1,
\end{align*}
where $\pi_1=1+\rho$, $\pi_2=\pi_3=-\rho$ and $\pi_4=\rho$.
In other words, if $Z_i$ is a continuous positive random variable having  {\color{black} CDF $F_{Z_i}(z)$}, $0<z<1$, $i=1,2,3,4$, where  $X~\sim{\rm Gamma}(\alpha,\theta)$ and $Y~\sim{\rm Gamma}(\beta,\theta)$, then {\color{black} \(Z\) admits a signed mixture representation; that is, its CDF can be written as the following linear combination:
\begin{align}\label{rep-stoch}
   	F_Z(z;\boldsymbol{\theta})
	=
	F_{{X\over X+Y}}(z)
	=
	\sum_{i=1}^{4} \pi_iF_{Z_i}(z),  
	\quad 0<z<1.
\end{align}
Furthermore, note that the above representation corresponds to a proper convex mixture only when \(\rho = 0\).
}
		
		\begin{proposition}\label{prop-1}
			Suppose that 
			\begin{itemize}
				\item[(a)] $X_1$ and $Y_1$ are independent variables having densities $f_X(x)$ and $f_Y(x)$, respectively;
				\item[(b)] $X_2$ and $Y_2$ are independent variables having densities $g_X(x)$ and $f_Y(x)$, respectively; 
				\item[(c)]  $X_3$ and $Y_3$ are independent variables having densities $f_X(x)$ and $g_Y(x)$, respectively;
				\item[(d)]  $X_4$ and $Y_4$ are independent variables having densities $g_X(x)$ and $g_Y(x)$, respectively.
			\end{itemize}
			If $Z_i$ has PDF $f_i(z)$, $i=1,2,3,4$, given in \eqref{def-f-1}, then
			\begin{align*}
				Z_i={X_i\over X_i+Y_i}, \quad i=1,2,3,4.
			\end{align*}
		\end{proposition}
		\begin{proof}
			The proof follows by using the Jacobian method and then by
			determining the marginal with respect to $Z_i$ of the random vector $(Z_i,U_i)^\top$, where $Z_i = X_i/(X_i + Y_i)$ and $U_i = Y_i$, $i=1,2,3,4$.
		\end{proof}
		
		By combining \eqref{rep-stoch} and Proposition \ref{prop-1}, with $(X, Y)^\top$ following a Morgenstern's bivariate  distribution  (with gamma margins), we get
        {\color{black}
        \begin{align}\label{rep-stoch-1}
   	F_Z(z;\boldsymbol{\theta})
	=
	F_{{X\over X+Y}}(z)
	=
	\sum_{i=1}^{4} \pi_iF_{{X_i\over X_i+Y_i}}(z),  
	\quad 0<z<1.
\end{align}
}


	\subsection{The cumulative distribution function}\label{subsec:cdf}
\noindent

Using the stochastic representation of $Z\sim {\rm EB}(\boldsymbol{\theta})$, given in \eqref{rep-stoch-1-1}, by the Law of total expectation we have (for $s=1/z-1$ and  $0<z<1$)
\begin{align}\label{id-cdf}
	F_Z(z;\boldsymbol{\theta})
	=
	\mathbb{P}\left({X\over X+Y}\leqslant z\right)
	=
	\mathbb{P}\left(s X\leqslant Y\right)
	&=
	1-
	\mathbb{P}\left(Y\leqslant s X\right)
	\nonumber
	\\[0,2cm]
	&=
	1
	-
	\int_0^\infty 
	\mathbb{P}\left(Y\leqslant s X\vert X=x \right)
	f_X(x)
	{\rm d}x
	\nonumber
		\\[0,2cm]
	&=
	1
	-
	\int_0^\infty 
	\int_0^{sx}
	f_{X,Y}(x,y)
	{\rm d}y
	{\rm d}x.
\end{align}

Since $(X, Y)^\top$ has Morgenstern's bivariate  distribution \eqref{Morgenstern-cdf}-\eqref{Morgenstern} (with gamma margins) (see Subsection 	\ref{EB model arising as a ratio}), the above identities are briefly written as
\begin{align}\label{int-1}
	F_Z(z;\boldsymbol{\theta})
	=
	1
-
\int_0^\infty
	f_X(x)
	\left\{
	F_Y(sx)+\rho[2F_X(x)-1]
	\left[\int_0^{sx}2f_Y(y)F_Y(y){\rm d}y-F_Y(sx)\right]
	\right\}
	{\rm d}x.
\end{align}

As
$Y~\sim{\rm Gamma}(\beta,\theta)$ {\color{black} and ${\rm d}F_Y^2(y)=2f_Y(y)F_Y(y){\rm d}y$}, 
we obtain
\begin{align}\label{int-interm}
	\int_0^{sx}2f_Y(y)F_Y(y){\rm d}y
    {\color{black}
    \, =
   F_Y^2(sx) 
   =
   \left[{\gamma (\beta ,\theta s x)\over \Gamma(\beta)}\right]^2.
    }
\end{align}

As $X~\sim{\rm Gamma}(\alpha,\theta)$ and $Y~\sim{\rm Gamma}(\beta,\theta)$, by using \eqref{int-interm}, $F_Z(z;\boldsymbol{\theta})$ in \eqref{int-1} is written as
%
\begin{align*}
F_Z(z;\boldsymbol{\theta})
&
=
1-
	{\theta ^{\alpha}\over\Gamma(\alpha)} 
	\int_0^\infty
	x^{\alpha -1}
	\exp(-\theta x)
	\left[
	{\gamma(\beta,\theta sx)\over \Gamma(\beta)}
	+
	\rho\left\{2\, {\gamma(\alpha,\theta x)\over \Gamma(\alpha)}-1\right\}
{\color{black}
    {\gamma(\beta,\theta sx)\over \Gamma(\beta)}
	\left\{
	{\gamma(\beta,\theta sx)\over \Gamma(\beta)}
    -
    1
	\right\}
    }
	\right]
	{\rm d}x
\\[0.2cm]
& {\color{black}
=
1
+
(\rho-1)\,
\frac{\theta^\alpha}{\Gamma(\alpha)\Gamma(\beta)}
\int_0^\infty
x^{\alpha-1} \exp(-\theta x)
\gamma(\beta,\theta s x)\,{\rm d}x
}
\\[0.2cm]
&
{\color{black}
-
\rho \,
\frac{2\theta^\alpha}{\Gamma^2(\alpha)\Gamma^2(\beta)}
\int_0^\infty
x^{\alpha-1} \exp(-\theta x)
\gamma(\alpha,\theta x)
\gamma^2(\beta,\theta s x)\,{\rm d}x
}
\\[0.2cm]
&
{\color{black}
+
\rho \,
\frac{2\theta^\alpha}{\Gamma^2(\alpha)\Gamma(\beta)}
\int_0^\infty
x^{\alpha-1} \exp(-\theta x)
\gamma(\alpha,\theta x)
\gamma(\beta,\theta s x)\,{\rm d}x
}
\\[0.2cm]
&
{\color{black}
+
\rho \,
\frac{\theta^\alpha}{\Gamma(\alpha)\Gamma^2(\beta)}
\int_0^\infty
x^{\alpha-1} \exp(-\theta x)
\gamma^2(\beta,\theta s x)\,{\rm d}x,
\quad 
s=1/z-1.
}
\end{align*}
By using the identities in \eqref{prop-app-3-1} and
\eqref{Prudnikov2002-2},
we have (for $0<z<1$)
{\color{black}
\begin{align*}
F_Z(z;\boldsymbol{\theta})
&=
1
+
(\rho-1)\,
\frac{\Gamma(\alpha+\beta)} {\beta\,\Gamma(\alpha)\Gamma(\beta)} \, 
(1-z)^\beta z^{\alpha}
\,{}_2F_1\!\left(\alpha+\beta,1;\beta+1;1-z\right)
\\[0.2cm]
&
+
\rho \,
\frac{2\Gamma(2\alpha+\beta)}{\alpha\beta\,\Gamma^2(\alpha)\Gamma(\beta)} \, 
(1-z)^\beta z^{2\alpha}
\,
F_2\!\left(2\alpha+\beta,1,1;\alpha+1,\beta+1;z,1-z\right)
\\[0.2cm]
&
+
\rho \,
\frac{\Gamma(\alpha+2\beta)}{\beta^2\,\Gamma(\alpha)\Gamma^2(\beta)} \,
(1-z)^{2\beta}
\left(\frac{z}{2-z}\right)^{\alpha+2\beta}
\,
F_2\!\left(\alpha+2\beta,1,1;\beta+1,\beta+1;
\frac{1-z}{2-z},\frac{1-z}{2-z}\right)
\\[0.2cm]
&
-
\rho \,
\frac{2}{\Gamma^2(\alpha)\Gamma^2(\beta)}
\int_0^\infty
u^{\alpha-1}\exp({-u})
\gamma(\alpha,u)
\gamma^2\!\left(\beta,\frac{1-z}{z}\,u\right) {\rm d}u.
\end{align*}
}
%
By using the well-known identity
\begin{align*}
B_x(a,b)={x^a(1-x)^b\over a}\, _2F_1(1,a+b;1+a;x),
\end{align*}
where $B_x(a,b)=\int_0^x t^{a-1}(1-t)^{b-1}{\rm d} t$ is the incomplete beta function,
we write $F_Z(z{\color{black} ;\boldsymbol{\theta}})$ as follows (for $0<z<1$)
\begin{align}\label{cdf-unit}
	F_Z(z;\boldsymbol{\theta})
	&=
	1
	+
	(\rho-1)
	{I_{1-z}(\beta,\alpha)} 
	\nonumber
	\\[0,2cm]
	&
    %
{\color{black}
+
\rho \,
\frac{2\Gamma(2\alpha+\beta)}{\alpha\beta\,\Gamma^2(\alpha)\Gamma(\beta)} \, 
(1-z)^\beta z^{2\alpha}
\,
F_2\!\left(2\alpha+\beta,1,1;\alpha+1,\beta+1;z,1-z\right)
}
\nonumber
\\[0.2cm]
&
{\color{black}
+
\rho \,
\frac{\Gamma(\alpha+2\beta)}{\beta^2\,\Gamma(\alpha)\Gamma^2(\beta)} \,
(1-z)^{2\beta}
\left(\frac{z}{2-z}\right)^{\alpha+2\beta}
\,
F_2\!\left(\alpha+2\beta,1,1;\beta+1,\beta+1;
\frac{1-z}{2-z},\frac{1-z}{2-z}\right)
}
\nonumber
\\[0.2cm]
&
{\color{black}
-
\rho \,
\frac{2}{\Gamma^2(\alpha)\Gamma^2(\beta)}
\int_0^\infty
u^{\alpha-1}\exp({-u})
\gamma(\alpha,u)
\gamma^2\!\left(\beta,\frac{1-z}{z}\,u\right) {\rm d}u,
}
\end{align}
where $I_z(a,b)=B_{z}(a,b)/\textrm{B}(a,b)$ is the  regularized incomplete beta function. {\color{black}To the best of our knowledge, the above integral does not admit a closed-form expression in terms of known special functions; therefore, it must be evaluated numerically.}

\begin{remark}\label{remark:beta_particular}
	When $\rho=0$, by using the well-known identity 
	$
		{I_{z}(\alpha,\beta)} 
		=
		1-
	{I_{1-z}(\beta,\alpha)},
$
	we have (for $0<z<1$)
	\begin{align*}
	F_Z(z;\boldsymbol{\theta})
=
{I_{z}(\alpha,\beta)},
	\end{align*}
which is a very well-known result in the literature.
\end{remark}

\subsection{Moment-generating function and moments}
\noindent

Given that $Z\sim {\rm EB}(\boldsymbol{\theta})$  has unitary support, its moment-generating function (MGF), $M_Z(t)$, and raw moments, $\mu_n$, $n\in\mathbb{N}$,  exist and are finite, particularly within an open interval encompassing zero. Then  $M_Z(t)$ may be expanded in a
power series about $0$ in the following form
\begin{align*}
M_Z(t)=1+\sum_{n=1}^{\infty} {\mu_n\over n!}\, t^n.
\end{align*}

In what follows, we derive a closed-form expression for  $\mu_n$. Indeed, by using the well-known formula for moments of positive order $\mathbb{E}(X^p)=p\int_{0}^{\infty} x^{p-1}\mathbb{P}(X>x){\rm d}x$ of positive random variables {\color{black}(which simplifies to $\mathbb{E}(X^p)=p\int_{0}^{1} x^{p-1}\mathbb{P}(X>x){\rm d}x$ for random variables with support $(0,1)$, since in this case $\mathbb{P}(X>x)=0$ for $x\geqslant 1$)}, we have
\begin{align*}
	\mu_n
	=
	n\int_{0}^{1} z^{n-1}[1-F_Z(z;\boldsymbol{\theta})]{\rm d}z,
\end{align*}
where $F_Z(z;\boldsymbol{\theta})$ is as in \eqref{cdf-unit}. So, from \eqref{cdf-unit}, 
\begin{align}\label{raw-moments}
	\mu_n
	&=
	(1-\rho)n\int_{0}^{1} z^{n-1} {I_{1-z}(\beta,\alpha)} {\rm d}z
	\nonumber
	\\[0,2cm]
	&
    %
    {\color{black}
-
\rho \,
\frac{2n \Gamma(2\alpha+\beta)}{\alpha\beta\,\Gamma^2(\alpha)\Gamma(\beta)} \, 
\int_0^1
(1-z)^\beta z^{2\alpha+n-1}
\,
F_2\!\left(2\alpha+\beta,1,1;\alpha+1,\beta+1;z,1-z\right)
{\rm d}z
}
\nonumber
\\[0.2cm]
&
{\color{black}
-
\rho \,
\frac{n\Gamma(\alpha+2\beta)}{\beta^2\,\Gamma(\alpha)\Gamma^2(\beta)} \,
\int_0^1
(1-z)^{2\beta}
\left(\frac{z}{2-z}\right)^{\alpha+2\beta}
z^{n-1}
\,
F_2\!\left(\alpha+2\beta,1,1;\beta+1,\beta+1;
\frac{1-z}{2-z},\frac{1-z}{2-z}\right)
{\rm d}z
}
\nonumber
\\[0.2cm]
&
{\color{black}
+
\rho \,
\frac{2n}{\Gamma^2(\alpha)\Gamma^2(\beta)}
\int_0^1
z^{n-1}
\int_0^\infty
u^{\alpha-1}\exp({-u})
\gamma(\alpha,u)
\gamma^2\!\left(\beta,\frac{1-z}{z}\,u\right) {\rm d}u
{\rm d}z.
}
\end{align}
For the particular scenario where $\alpha=\beta$, we obtain $\mu_1=1/2$ (see Subsection \ref{Symmetry}). Due to the integral's complexity in \eqref{raw-moments}, numerical methods are requisite for its computation.
{\color{black} Since closed-form moments are unavailable, Table \ref{tab:moments} provides numerical approximations for $\mu_n$ (with $n=1,\ldots,10,50,100$) for different choices of $\rho$ in the special case when $\alpha=\beta=1$.

\begin{table}[H]
	\caption{{\color{black} Moments values for different choices of $\rho$ ($\alpha=\beta=1)$.}}
\centering
\begin{tabular}{cccccccccc}
  \hline
$\rho$ & $-0.95$ & $-0.9$ & $-0.75$ & $-0.5$ & $0$ & $0.5$ & $0.75$ & $0.9$ & $0.95$ \\ 
  \hline
$\mu_1$ & 0.500 & 0.500 & 0.500 & 0.500 & 0.500 & 0.500 & 0.500 & 0.500 & 0.500 \\ 
  $\mu_2$ & 0.348 & 0.347 & 0.345 & 0.341 & 0.333 & 0.326 & 0.322 & 0.319 & 0.319 \\ 
  $\mu_3$ & 0.272 & 0.271 & 0.268 & 0.262 & 0.250 & 0.238 & 0.232 & 0.229 & 0.228 \\ 
  $\mu_4$ & 0.225 & 0.224 & 0.220 & 0.213 & 0.200 & 0.187 & 0.180 & 0.176 & 0.175 \\ 
  $\mu_5$ & 0.193 & 0.192 & 0.188 & 0.181 & 0.167 & 0.153 & 0.146 & 0.142 & 0.140 \\ 
  $\mu_6$ & 0.169 & 0.168 & 0.164 & 0.157 & 0.143 & 0.129 & 0.122 & 0.118 & 0.116 \\ 
  $\mu_7$ & 0.151 & 0.150 & 0.146 & 0.139 & 0.125 & 0.111 & 0.104 & 0.100 & 0.099 \\ 
  $\mu_8$ & 0.137 & 0.135 & 0.131 & 0.124 & 0.111 & 0.098 & 0.091 & 0.087 & 0.086 \\ 
  $\mu_9$ & 0.125 & 0.123 & 0.119 & 0.113 & 0.100 & 0.087 & 0.081 & 0.077 & 0.075 \\ 
  $\mu_{10}$ & 0.115 & 0.113 & 0.110 & 0.103 & 0.091 & 0.078 & 0.072 & 0.068 & 0.067 \\ 
  $\mu_{50}$ & 0.028 & 0.027 & 0.026 & 0.024 & 0.020 & 0.015 & 0.013 & 0.012 & 0.011 \\ 
  $\mu_{100}$ & 0.014 & 0.014 & 0.013 & 0.012 & 0.010 & 0.008 & 0.006 & 0.006 & 0.006 \\ 
   \hline
\end{tabular}
	\label{tab:moments}
\end{table}
}

\section{Estimation}\label{esti}
\noindent

{\color{black}
\subsection{Maximum likelihood estimation}
}
\noindent

The log-likelihood function for $\boldsymbol{\theta} = (\alpha,\beta,\rho)^\top$ based on a sample of $n$ independent observations is given by
\begin{align}\label{loglik}
\ell(\boldsymbol{\theta})
&=
\sum_{i=1}^{n}
\log
\Bigg[
(1+\rho) \,
\frac{z_i^{\alpha-1}(1-z_i)^{\beta-1}}
{\mathrm{B}(\alpha,\beta)}
\nonumber 
\\[0,2cm] \nonumber
&\quad
-
\rho \,
\frac{2\Gamma(2\alpha+\beta)}
{\alpha\Gamma^2(\alpha)\Gamma(\beta)}
\frac{z_i^{2\alpha-1}(1-z_i)^{\beta-1}}
{(1+z_i)^{2\alpha+\beta}}
\,
_2F_1\!\left(1,2\alpha+\beta;1+\alpha;\frac{z_i}{1+z_i}\right)
\\[0,2cm] \nonumber
&\quad
-
\rho \,
\frac{2\Gamma(\alpha+2\beta)}
{\beta\Gamma(\alpha)\Gamma^2(\beta)}
\frac{z_i^{\alpha-1}(1-z_i)^{2\beta-1}}
{(2-z_i)^{\alpha+2\beta}}
\,
_2F_1\!\left(1,\alpha+2\beta;1+\beta;\frac{1-z_i}{2-z_i}\right)
\\[0,2cm] 
&\quad
+
\rho \,
\frac{4^{1-\alpha-\beta}\Gamma(2\alpha+2\beta)}
{\alpha\beta\Gamma^2(\alpha)\Gamma^2(\beta)}
z_i^{2\alpha-1}(1-z_i)^{2\beta-1} \,
F_2\!\left(
2\alpha+2\beta,1,1;
\alpha+1,\beta+1;
\frac{z_i}{2},
\frac{1-z_i}{2}
\right)
\Bigg]. 
\end{align}

The maximum likelihood estimator (MLE)
$\boldsymbol{\widehat{\theta}} = (\widehat{\alpha},\widehat{\beta},\widehat{\rho})^\top$ of $\boldsymbol{\theta} = (\alpha,\beta,\rho)^\top$
is obtained by the maximization of the log-likelihood function (\ref{loglik}). However, it is
not possible to derive an analytical solution for the MLE
$\boldsymbol{\widehat{\theta}}$; therefore, we must require
a numerical solution using some optimization algorithm such as Newton–Raphson
and quasi-Newton. But, computations of likelihood $\ell(\boldsymbol{\theta})$ can be extremely costly due to the dependence of PDF \eqref{pdf-main} and CDF \eqref{cdf-unit} on special functions such as the Gauss hypergeometric ${}_2F_1$ and Appell $F_2$. Moreover, optimization methods may not converge properly.

{\color{black}
\begin{remark}
The Fisher information matrix associated with the parameter vector 
\(\boldsymbol{\theta}=(\alpha,\beta,\rho)^\top\) is defined by
\begin{equation*}
    I(\boldsymbol{\theta})=\mathbb E\left[
\nabla \log f_Z(Z;\boldsymbol{\theta})\,
(\nabla \log f_Z(Z;\boldsymbol{\theta}))^\top
\right],
\end{equation*}
where $\nabla=(\partial/\partial \alpha, \partial/\partial \beta, \partial/\partial \rho)^\top$ is the gradient operator and $f_Z(z;\boldsymbol{\theta})$ is the EB PDF defined in \eqref{pdf-main}.

Observe that, when \(\rho\to 0\),
\[
f_Z(z;\boldsymbol{\theta})
\to
f_0(z;\alpha,\beta)
=
\frac{z^{\alpha-1}(1-z)^{\beta-1}}
{\mathrm B(\alpha,\beta)},
\quad 0<z<1,
\]
that is, the classical Beta distribution. Hence, the corresponding Fisher information matrix reduces to
\[
I_0(\alpha,\beta)
=
\begin{pmatrix}
\psi_1(\alpha)-\psi_1(\alpha+\beta)
&
-\psi_1(\alpha+\beta)
\\[0.2cm]
-\psi_1(\alpha+\beta)
&
\psi_1(\beta)-\psi_1(\alpha+\beta)
\end{pmatrix},
\]
where \(\psi_1(\cdot)\) denotes the trigamma function.


On the other hand, for the cases $0<\vert \rho\vert\leqslant 1$ is unknown, the Fisher information matrix is given by
\begin{equation}\label{eq:Fisher_information}
  I(\boldsymbol{\theta})
= \left( \int_0^1
\frac{\partial \log f_Z(z;\boldsymbol{\theta})}{\partial\theta_i}
\frac{\partial \log f_Z(z;\boldsymbol{\theta})}{\partial\theta_j}
f_Z(z;\boldsymbol{\theta}){\rm d}z \right)_{1\leq i,j\leq 3}, ~~\theta_1\equiv\alpha, \theta_2\equiv\beta, \theta_3\equiv\rho.
\end{equation}

Since the density \(f_Z(z;\boldsymbol{\theta})\) involves Gauss hypergeometric and Appell functions, the Fisher information matrix \(I(\boldsymbol{\theta})\) does not admit explicit closed-form expressions. Therefore, following \citet{EfronHinkley1978}, it is preferable to use the observed Fisher information, defined as the negative Hessian matrix, since it is directly computed from the log-likelihood function and typically provides more accurate standard error approximations in finite samples.
\end{remark}
}

{\color{black}
\subsection{Two-stage pseudolikelihood}
}
\noindent

{\color{black} Aiming to overcome the computational challenges of the MLE, we provide two alternative estimation approaches. Our first} strategy is as follows: given a random sample from $(X, Y)^\top$, we first fit the margins as $\operatorname{Gamma} (\alpha, 1)$ and $\operatorname{Gamma} (\beta, 1)$, and then 
{\color{black}
we estimate $\rho$ in a second step, done using the two-stage pseudolikelihood (PML) estimation. In this case, given an observed sample $(x_i, y_i)$, $i=1,\cdots, n$, the PML of $\widehat{\rho}$ is given by:
\begin{equation}\label{mletheta}
    \widehat{\rho} = \arg\max_\rho\sum_{i=1}^n \log c(F(x_i), G(y_i); \rho),
\end{equation}
where $c(\cdot, \cdot;\rho)$ is the copula density and the maximization is taken over the respective parameter space of $\rho$.
Although this procedure does not guarantee the global optimum of the joint MLE, it offers computational advantages in terms of cost and provides a good alternative for obtaining parameter estimates (cf. \cite{lawless2011comparison, lima2024assessing}).}

Thus, comparing $X$ and $Y$ through $Z=X/(X+Y)$, we obtain a natural estimator for the parameters of $Z$.
Naturally, we can first use the estimates from the $\operatorname{Beta} (\alpha, \beta )$ model and then estimate $\rho$ in a second stage, thus improving the initial Beta fit. 
{\color{black} The Algorithm \ref{algo:beta_estimation} summarizes it in the two steps.
These steps are effective in scenarios where the Beta model provides a good initial fit to the data.

\begin{algorithm}[htb!]
\caption{Estimation of $\boldsymbol{\theta}$ based on improving Beta distribution}
\label{algo:beta_estimation}
\begin{algorithmic}[1]
    \Statex \textbf{Input:} a random sample $\mathbf{Z} = (Z_1, \ldots, Z_n)^\top \in [0,1]^n$.
    \Statex \textbf{Output:}  $\widehat{\boldsymbol{\theta}} = (\hat{\alpha}, \hat{\beta}, \hat{\rho})^\top$.

 \State Using the $\operatorname{Beta}(\alpha, \beta)$ estimators, compute
 $$(\hat{\alpha}, \hat{\beta})^\top = (\hat{\alpha}(\mathbf{Z}), \hat{\beta}(\mathbf{Z}))^\top;$$

\State  Estimate $\rho$ using an appropriate optimization algorithm:
\begin{equation}\label{eq:rho_est1}
        \widehat{\rho} \in \arg\max_{\rho\in[0,1]} \ell(\rho),
    \end{equation}
    where $\ell(\rho) := \ell(\boldsymbol{\theta})\mid_{(\alpha, \beta)=(\hat{\alpha}, \hat{\beta})}$;

  \State \textbf{Return} $\widehat{\boldsymbol{\theta}} = (\hat{\alpha}, \hat{\beta}, \hat{\rho})^\top$.
    
\end{algorithmic}
\end{algorithm}


When the unit data is obtained from $Z = {X}/{(X+Y)}$, then, using Theorem 2 in \cite{Lai1978}, we obtain the estimator
\begin{equation}\label{eq:rho_est2}
    \widehat{\rho} = \frac{\widehat{\operatorname{corr}}(X,Y) \hat{\sigma}_x \hat{\sigma}_y}{\hat{\beta}_x \hat{\beta}_y},
\end{equation}
where $\sigma_x = \operatorname{Var} X$, $\sigma_y = \operatorname{Var} Y$, $\beta_x = \int_0^\infty G(x; \alpha,1) \left[1-G(x; \alpha,1)\right] {\rm d}x $,
and 
$$\beta_y = \int_0^\infty G(x; \beta,1) \left[1-G(x; \beta,1)\right] {\rm d}x, $$
in which $G(\cdot; \theta, 1)$ denotes the CDF of a Gamma model with shape $\theta$ and scale 1. {\color{black} It is easy to see that
\begin{equation*}
     \widehat{\rho} = \frac{\pi}{\alpha\beta} \frac{ \widehat{\operatorname{cov}}(X,Y) \Gamma(\alpha+1)\Gamma(\beta+1)}{\Gamma(\alpha+1/2)\Gamma(\beta+1/2)}.
\end{equation*}
Consequently, the Algorithm \ref{algo:gamma_estimation} can be used. If necessary, $(X,Y)$ are rescaled to ensure that they share the same scale parameter.

\begin{algorithm}[htb!]
\caption{Estimation of $\boldsymbol{\theta}$ based on improving Beta distribution}
\label{algo:gamma_estimation}
\begin{algorithmic}[1]
    \Statex \textbf{Input:} a random sample $(\mathbf{X}, \mathbf{Y}) = ((X_1,Y_1)^\top, \ldots, (X_n, Y_n)^\top)^\top \in \left((0,\infty) \times (0,\infty)\right)^n$.
    \Statex \textbf{Output:}  $\widehat{\boldsymbol{\theta}} = (\hat{\alpha}, \hat{\beta}, \hat{\rho})^\top$.

 \State Using the $\operatorname{Gamma}(\alpha, \theta)$ estimators, compute the shape parameters
 $$(\hat{\alpha}, \hat{\beta})^\top = (\hat{\alpha}(\mathbf{X}), \hat{\beta}(\mathbf{Y}))^\top;$$

\State  Estimate $\rho$ using \eqref{mletheta} or \eqref{eq:rho_est2};

  \State \textbf{Return} $\widehat{\boldsymbol{\theta}} = (\hat{\alpha}, \hat{\beta}, \hat{\rho})^\top$.
    
\end{algorithmic}
\end{algorithm}

\begin{remark}
    The second step \eqref{mletheta} in Algorithm  \ref{algo:gamma_estimation} requires optimization procedures. A natural initial guest is \eqref{eq:rho_est2}.
    Using the Spearman correlation $\rho_S = \frac{\rho}{3}\in\left[-1/3, 1/3\right],$ 
    we can choose the starting values $\rho_0 \approx 3 \hat{\rho}_S$. Similarly, we could use the Kendall estimate ($\tau = 2/9 \rho \in [-2/9, 2/9] $).
\end{remark}

}}

\section{Simulation study}\label{sec:simulation}

\subsection{Random sample generator}
\noindent


{\color{black} From \eqref{Morgenstern-cdf}, the joint distribution function \(F_{X,Y}(x,y)\) is obtained by composing the copula \(C(\cdot,\cdot)\), defined in \eqref{eq:copula_Morgenstern}, with the marginal Gamma CDFs \(F_X(x)\) and \(F_Y(y)\), where \(X\sim \mathrm{Gamma}(\alpha,\theta)\) and \(Y\sim \mathrm{Gamma}(\beta,\theta)\). That is,} 
$$ F_{X,Y}(x,y) = C(F_X(x), F_Y(y)), ~~x,y>0.$$

Aiming to simulate a random sample $Z_1, \ldots, Z_n$ from the extended Beta PDF \eqref{pdf-main}, we use the Algorithm \ref{algo:random_sample}, which simulates $Z = {X}/{(X+Y)} \sim f_Z(\cdot;\boldsymbol{\theta})$ based on a simulation of $(X,Y)\sim F_{X,Y}$. For a further discussion on this algorithm, we refer the reader to \cite{mai2017simulating}.

\begin{algorithm}[htb!]
\caption{Computational sampling of the extended Beta distribution}
\label{algo:random_sample}
\begin{algorithmic}[1]
    \Statex \textbf{Input:} $\bm\theta = (\alpha, \beta, \rho)^\top \in (0,\infty)^2\times[-1,1]$.
    \Statex \textbf{Output:}  $Z\sim f_Z(\cdot; \bm\theta)$.

 \State Simulate $U_2 \sim U[0, 1]$.

\State  Compute the partial derivative 
$$F_{U_1 | U_2} (u_1) := \left. \frac{\partial}{\partial u_2} C(u_1, u_2)\right|_{u_2 = U_2} = u_1\left[ 1+\rho(1-u_1)(1-2u_2) \right],~ u_1 \in [0, 1].$$

\State  Compute the generalized inverse $F_{U_1|U_2}^{-1}$ numerically:
$$F_{U_1|U_2}^{-1} (v) := \inf\left\{u_1 > 0; F_{U_1|U_2} (u_1) \geq v\right\},~ v \in (0, 1),$$

\State  Get a sample $V \sim U[0, 1]$, independent of $U_2$.

\State  Define $U_1 := F_{U_1| U_2}^{-1} (V)$ and take $(U_1, U_2)$ as the random vector of the copula $C$.

\State Compute $(X,Y)\sim F_{X,Y}$ as
$$ (X, Y)^\top = (F_X^{-1}(U_1), F_Y^{-1}(U_2))^\top, $$
where $F_{X,Y}$ is given in \eqref{Morgenstern-cdf} and $F_X$ and $F_Y$ denote the CDFs of Gamma models as described in Subsection \ref{subsec:cdf}.

  \State \textbf{Return} $Z=\frac{X}{X+Y}$.
    
\end{algorithmic}
\end{algorithm}

\subsection{R implementation and required packages}\label{subsec:packages}
\noindent

For this study, all implementations were performed in R software \citep{Rsoftware}.

{
We implement the Appell function $F_2$ in R using the double series \eqref{eq:F2}. The convergence region of \eqref{eq:F2} is sufficient to implement PDF \eqref{pdf-main} and CDF \eqref{cdf-unit} of the EB model. For other function implementations, a wider range may be required. In such cases, numerical integral representations can be employed. Some additional formulas for the Appell function $F_2$ can be found in \cite{brychkov2014some} and \cite{Ananthanarayan2023}. 
 }
The generalized hypergeometric function is computed using \textsf{hyperg\_2F1} from the \textsf{gsl} package. Thus, the probability density function (PDF) \eqref{pdf-main} can be implemented in R.

Referring to the Algorithm \ref{algo:random_sample}, the numerical inverse is required. We achieve this by the function \textsf{uniroot}.

A gradient-based method can be applied to optimize the log-likelihood function. The gradient of the objective function is obtained using the function \textsf{grad} in the package \textsf{numDeriv}.

\begin{remark}
    Another simulation algorithm could be proposed by directly inverting the CDF \eqref{cdf-unit} and applying it to a standard uniform random variable.
\end{remark}

\subsection{Simulation results {\color{black}for the two-stage pseudolikelihood estimator}}\label{sec_simulations}
\noindent

To evaluate the performance of the {\color{black}  two-stage pseudolikelihood (PML)} estimator $\hat{\bm \theta}$, we fix several values of the parameters $\bm\theta=(\alpha,\beta, \rho)^\top$, and then generate $N=1,000$ Monte-Carlo samples with sample size $n\in\{30, 100, 300, 1000\}$, of the random variable $Z\sim {\rm EB}(\alpha,\beta, \rho)$.
We analyze the {\color{black} PML estimates $\hat{\bm \theta}$ (Algorithm \ref{algo:gamma_estimation})} through relative bias (RB), and root mean squared error (RMSE).

We use the gradient descent algorithm to achieve the minimum of the minus log-likelihood function \cite[cf.][]{shalev2014understanding} to obtain the {\color{black}PML} estimates for the EB distribution. 
The gradient of the objective function can be obtained numerically, as described in Subsection \ref{subsec:packages}.

For the simulation, the following procedure was carried out:

\begin{algorithm}[H]
\caption{Assessing {\color{black}PML} estimates via Monte Carlo simulation}
\label{algo:mc}
\begin{algorithmic}[1]
    \Statex \textbf{Input:} $\bm\theta = (\alpha, \beta, \rho)^\top \in (0,\infty)^2\times[-1,1]$, $n$ (sample size), $N$ (number of Monte Carlo simulations)
    \Statex \textbf{Output:}  RB and RMSE of ML estimates.

\For{$j = 1$ to $N$}
   \State Simulate a Monte-Carlo sample
     $\mathbf{Z}^{(j)}=(Z_1^{(j)}, \cdots, Z_n^{(j)})$ via Algorithm \ref{algo:random_sample}.

\State  Compute $\hat{\boldsymbol{\theta}}^{(j)} = \hat{\boldsymbol{\theta}}^{(j)}(\mathbf{Z}^{(j)})$, {\color{black} using Algorithm \ref{algo:gamma_estimation}.}

 \EndFor

\State Compute the RB and RMSE as follows:
     $$RB(\hat{\theta}_k) = \frac{1}{N}\sum_{j=1}^N \frac{\hat{\theta}^{(j)}_k - \theta_k}{\theta_k}, ~~k=1,2,3,~\mbox{and}~ \theta_1=\alpha, \theta_2=\beta, \theta_3=\rho, $$
     and 
     $$RMSE(\hat{\theta}_k) = \frac{1}{N}\sum_{j=1}^N \left(\hat{\theta}^{(j)}_k - \theta_k\right)^2.$$

  \State \textbf{Return} $RB(\hat{\theta}_k)$ and $RMSE(\hat{\theta}_k)$ for $k=1,2,3,~\mbox{and}~ \theta_1=\alpha, \theta_2=\beta, \theta_3=\rho$.
    
\end{algorithmic}
\end{algorithm}

The simulation results are presented in Figures \ref{fig:sim_neg} and \ref{fig:sim_posi}.
The {\color{black}PML} estimator shows good behavior with minimal RB and low RMSE. Furthermore, it is clear that increasing the sample size $n$ leads RB and RMSE to approach zero.
This demonstrates the precision of the {\color{black}PML} estimator for the UF distribution, even with small sample sizes. Additionally, we observe that the RMSE of $\rho$ estimates is less affected by variations in sample size.
In general, for $n \geq 100$, the RB is below 0.10 and converges to zero as $n$ increases.  
Although relatively low (in absolute values), the relative bias of $\rho$ exhibited a different sign compared to the estimates of the parameters $\alpha$ and $\beta$.

Regarding the RMSE, large values were observed in some cases for small sample sizes. However, the value decreased when larger samples were taken ($n\in\{300, 1000\}$).  
The RMSE of the $\beta$ estimates was high ($\geq 0.5$) when $n = 30$, in both cases $\rho = -0.5$ and $\rho = 0.5 $. When $n = 100$, the values decreased to 0.20, stabilizing below 0.10 only for $n \geq 300$.  

In Figure \ref{fig:hist_MC_results}, we fixed $\bm\theta^\top = (2, 3, 0.5)$ and plotted the histogram of Monte Carlo simulation results of estimates. When $n=30$ (top), the estimates of the three parameters were right-skewed.  
We can see that as $n$ increases, the figures become symmetric, showing the good asymptotic behavior of the estimates. The same behavior occurs with the other parameters studied.

\begin{figure}[H]
	\centering
        \includegraphics[width=0.9\linewidth, height=7cm]{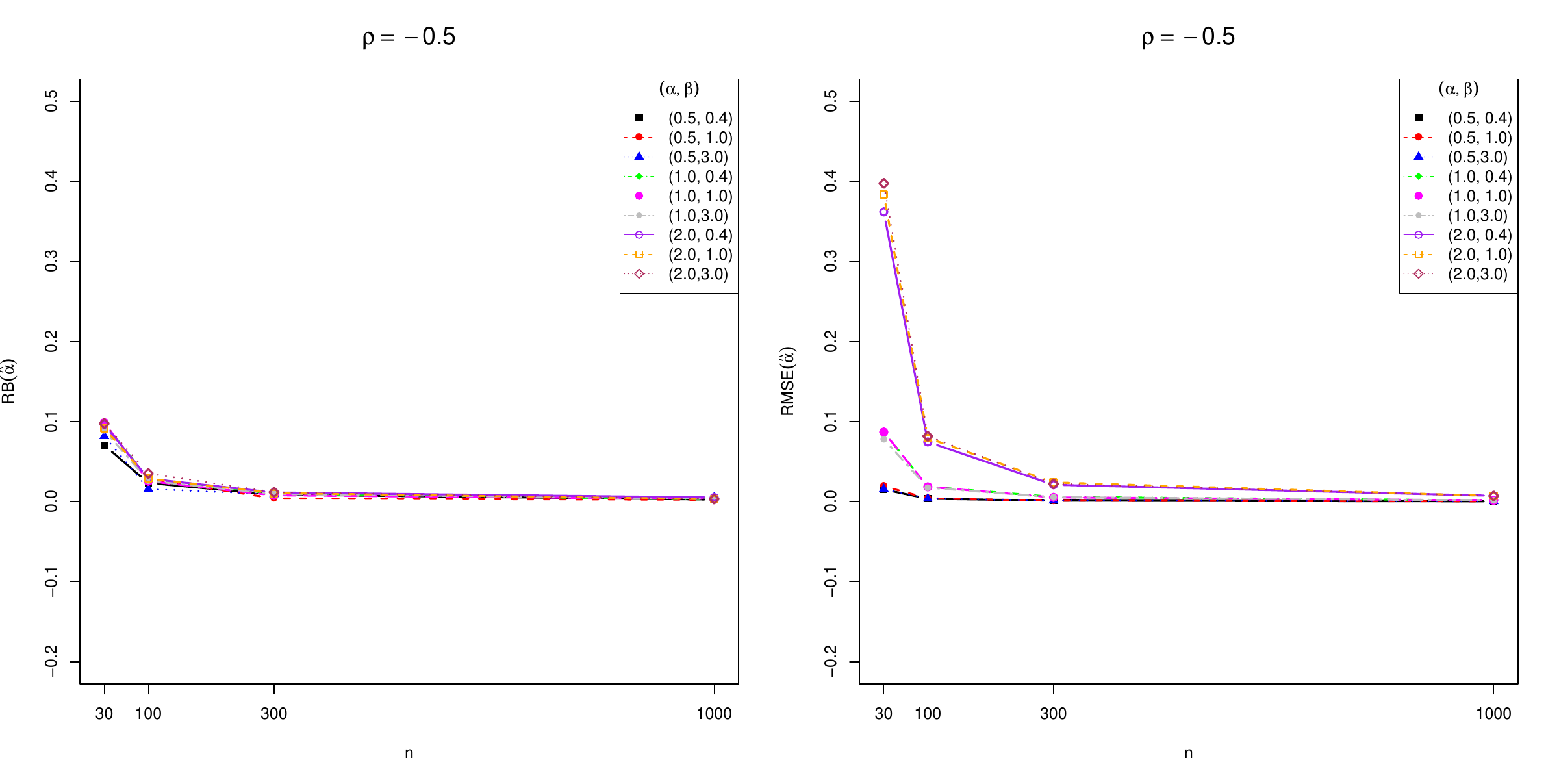}
            \includegraphics[width=0.9\linewidth, height=7cm]{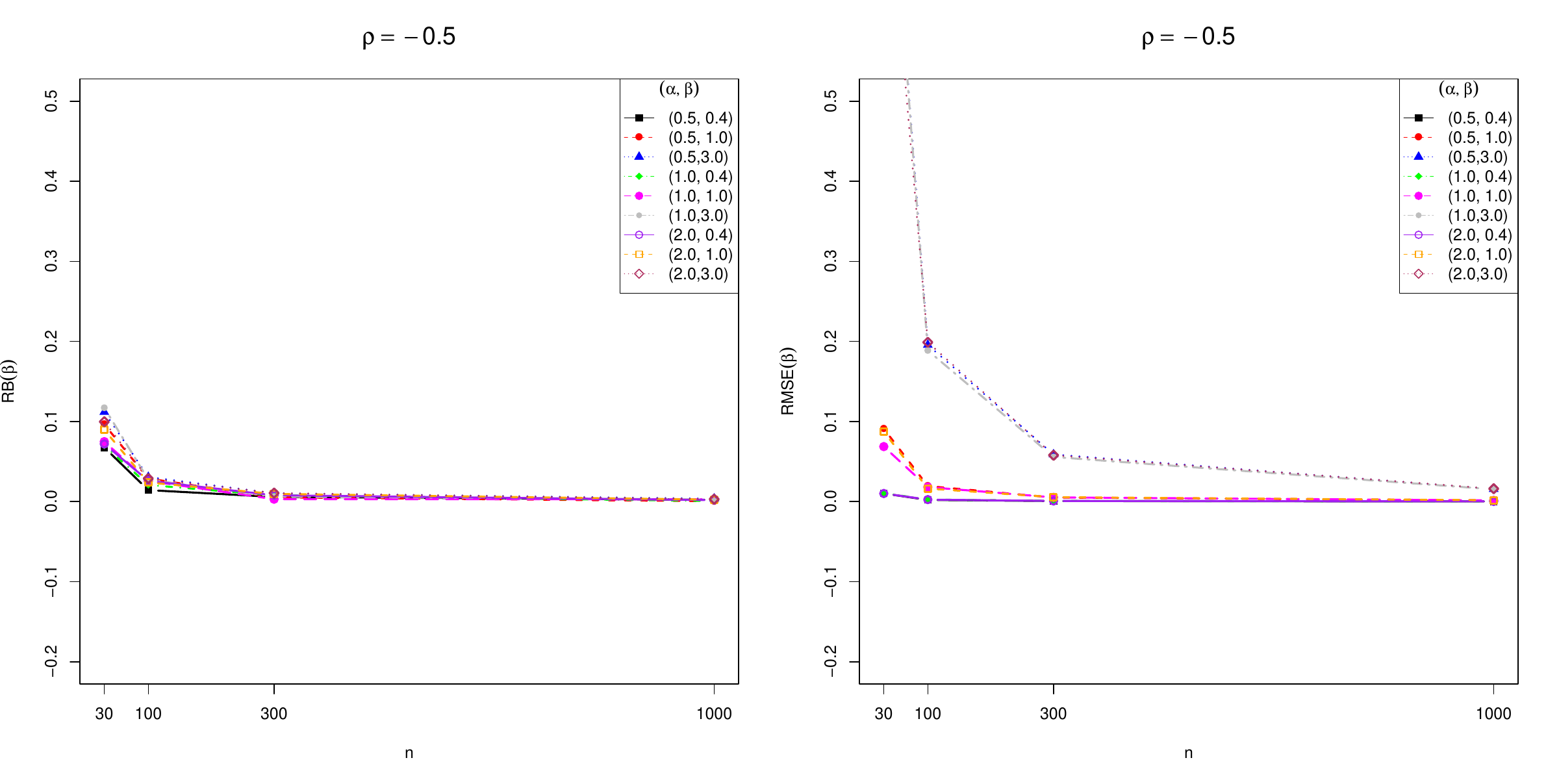}
	\includegraphics[width=0.9\linewidth, height=7cm]{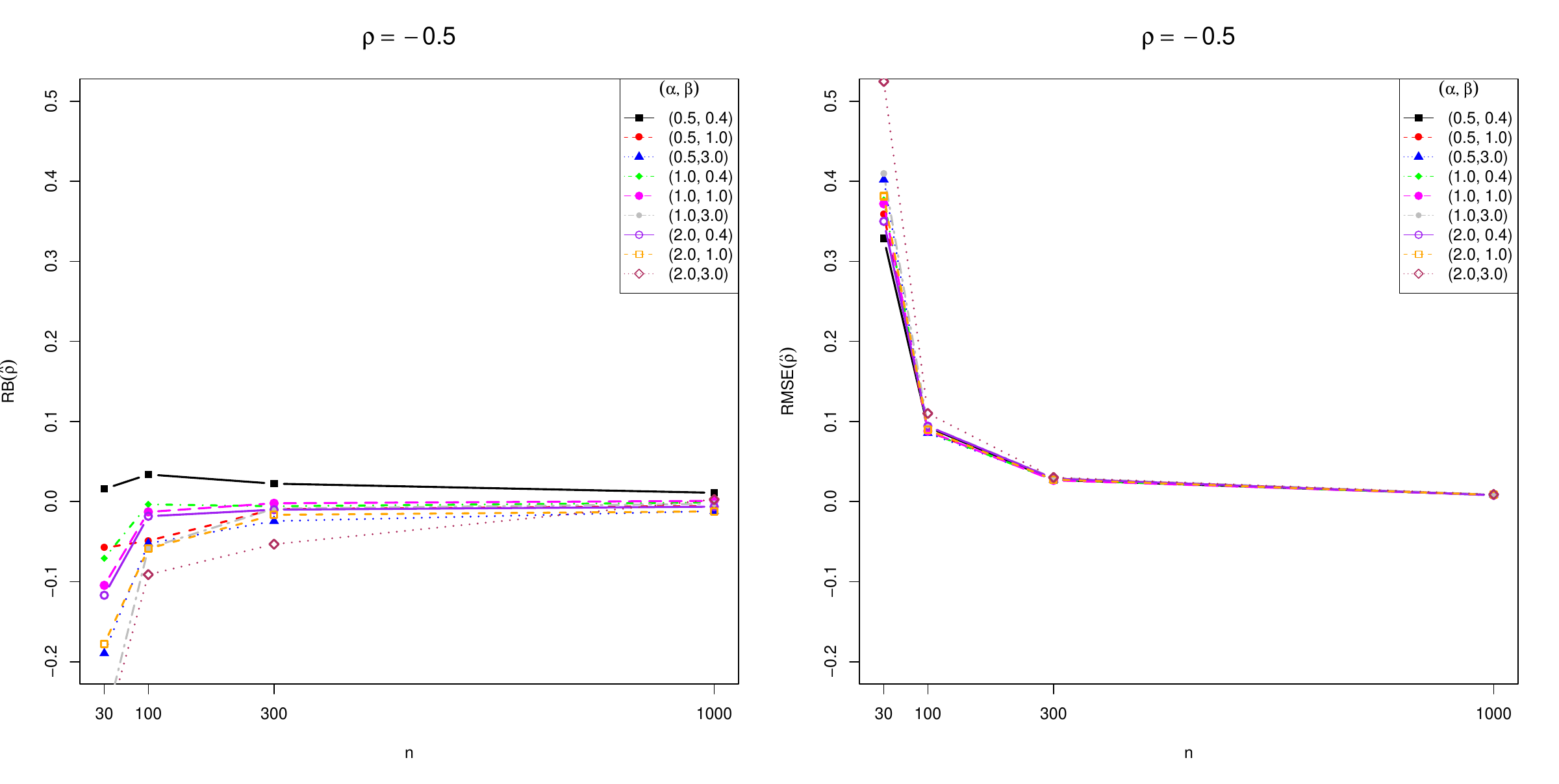}
	\caption{RB (left) and RMSE (right) for the {\color{black}PML} estimates of $\alpha$ (top), $\beta$ (middle), and $\rho$ (bottom), with varying $(\alpha, \beta)$ and fixing negative $\rho=-0.5$.}
	\label{fig:sim_neg}
\end{figure}

\begin{figure}[H]
	\centering
        \includegraphics[width=0.9\linewidth, height=7cm]{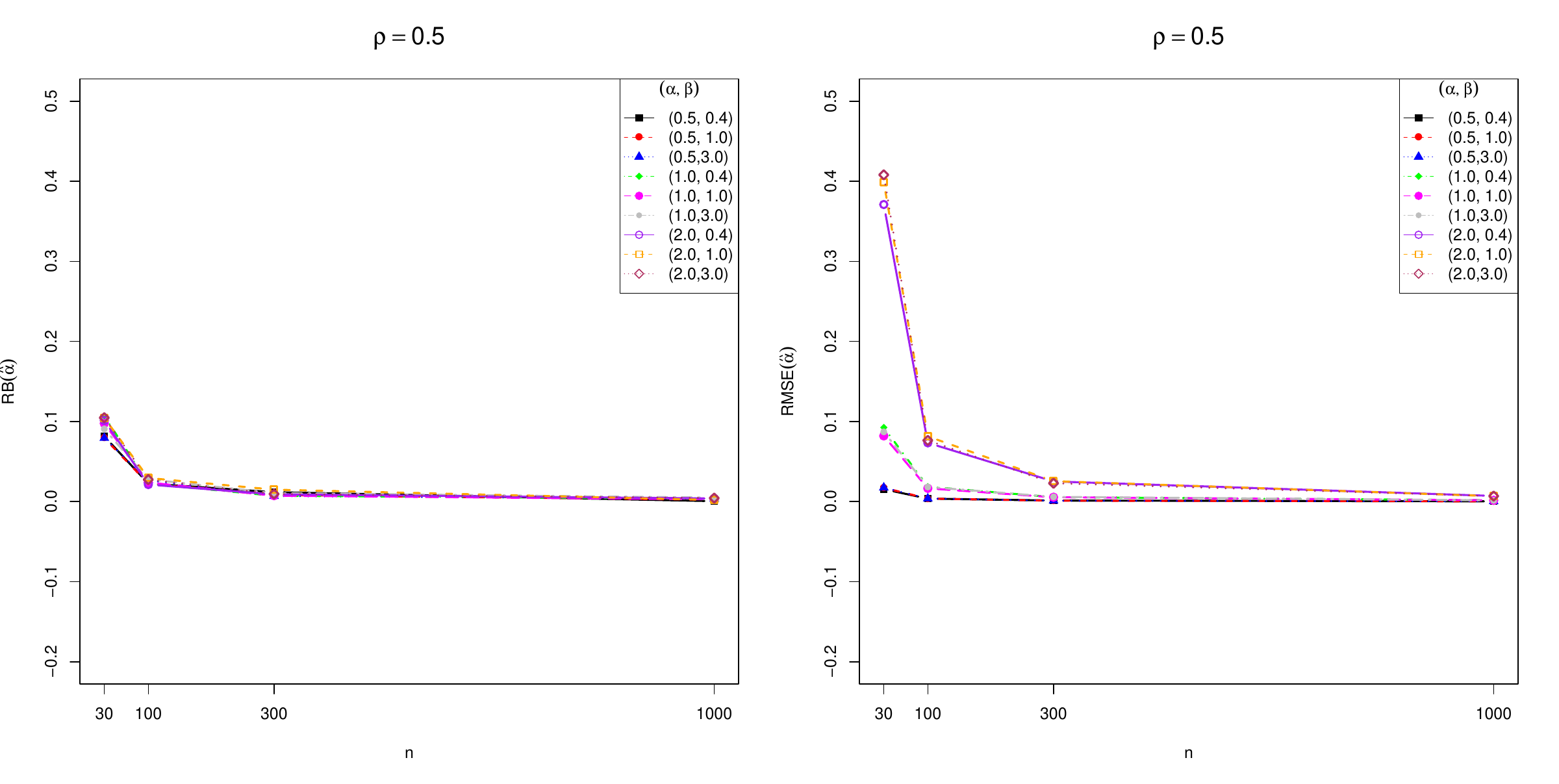}
            \includegraphics[width=0.9\linewidth, height=7cm]{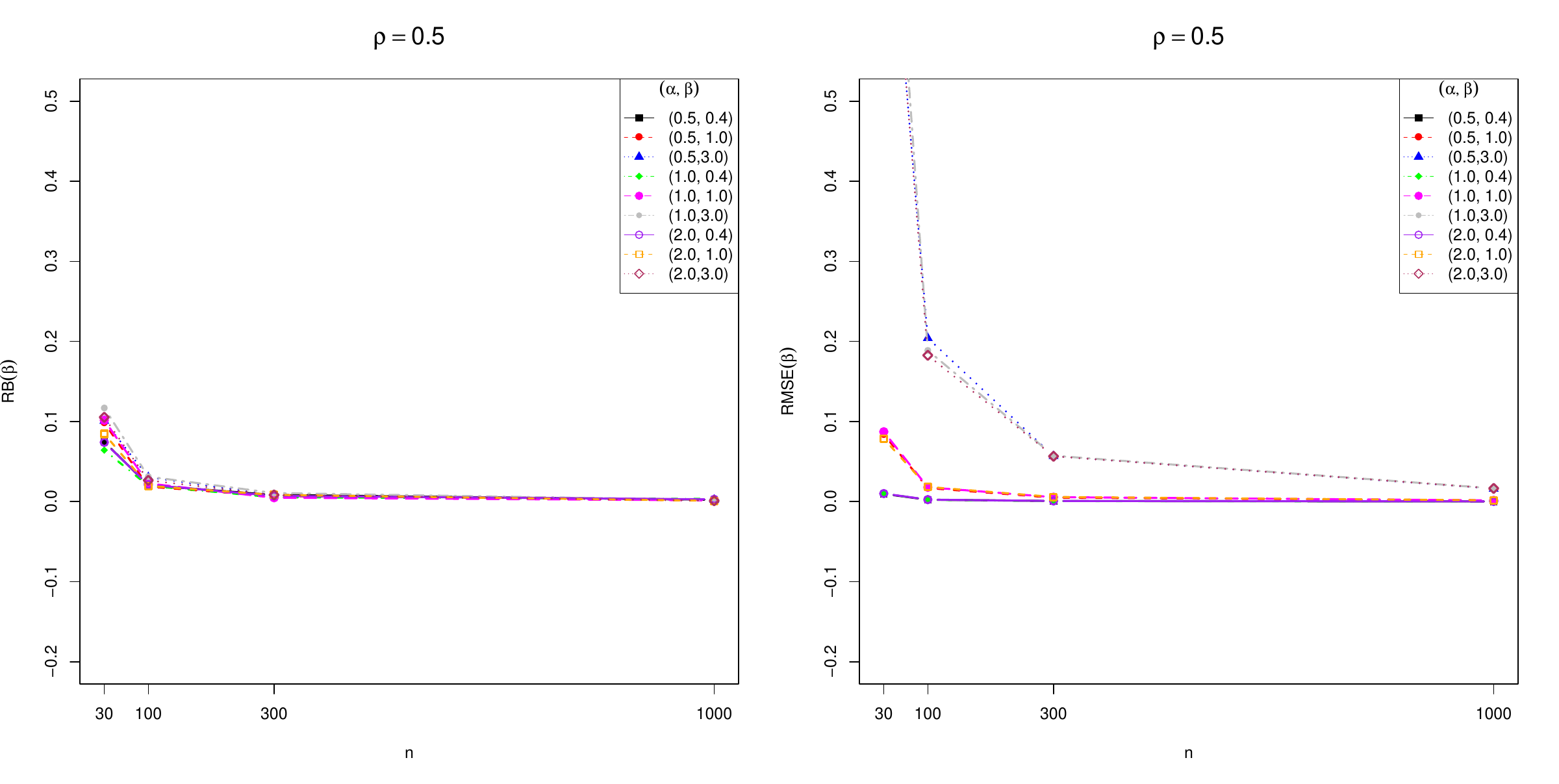}
	\includegraphics[width=0.9\linewidth, height=7cm]{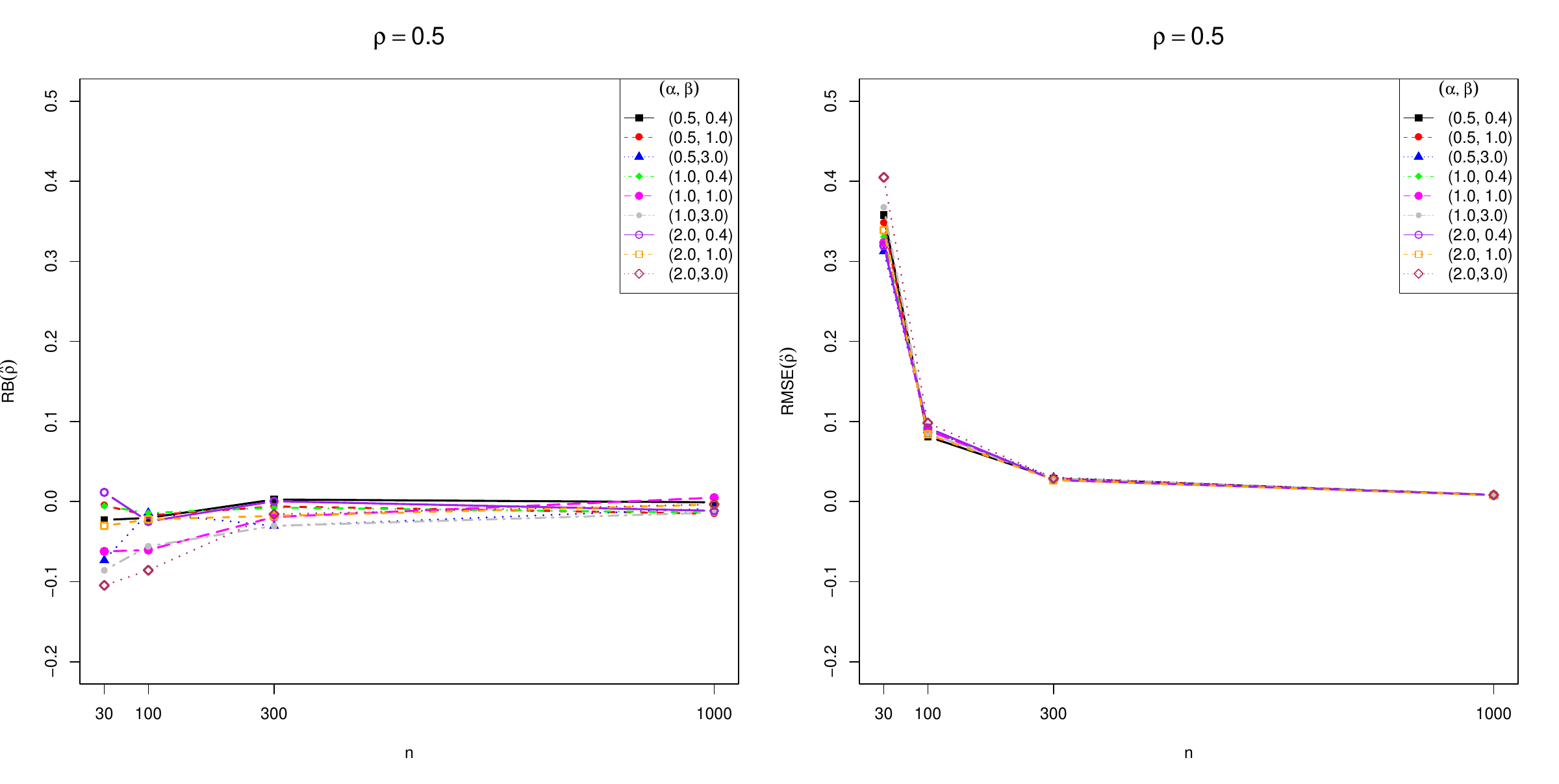}
	\caption{RB (left) and RMSE (right) for the {\color{black}PML} estimates of $\alpha$ (top), $\beta$ (middle), and $\rho$ (bottom), with varying $(\alpha, \beta)$ and fixing positive $\rho=0.5$.}
	\label{fig:sim_posi}
\end{figure}

\begin{figure}[H]
	\centering
        \includegraphics[width=0.9\linewidth, height=5cm]{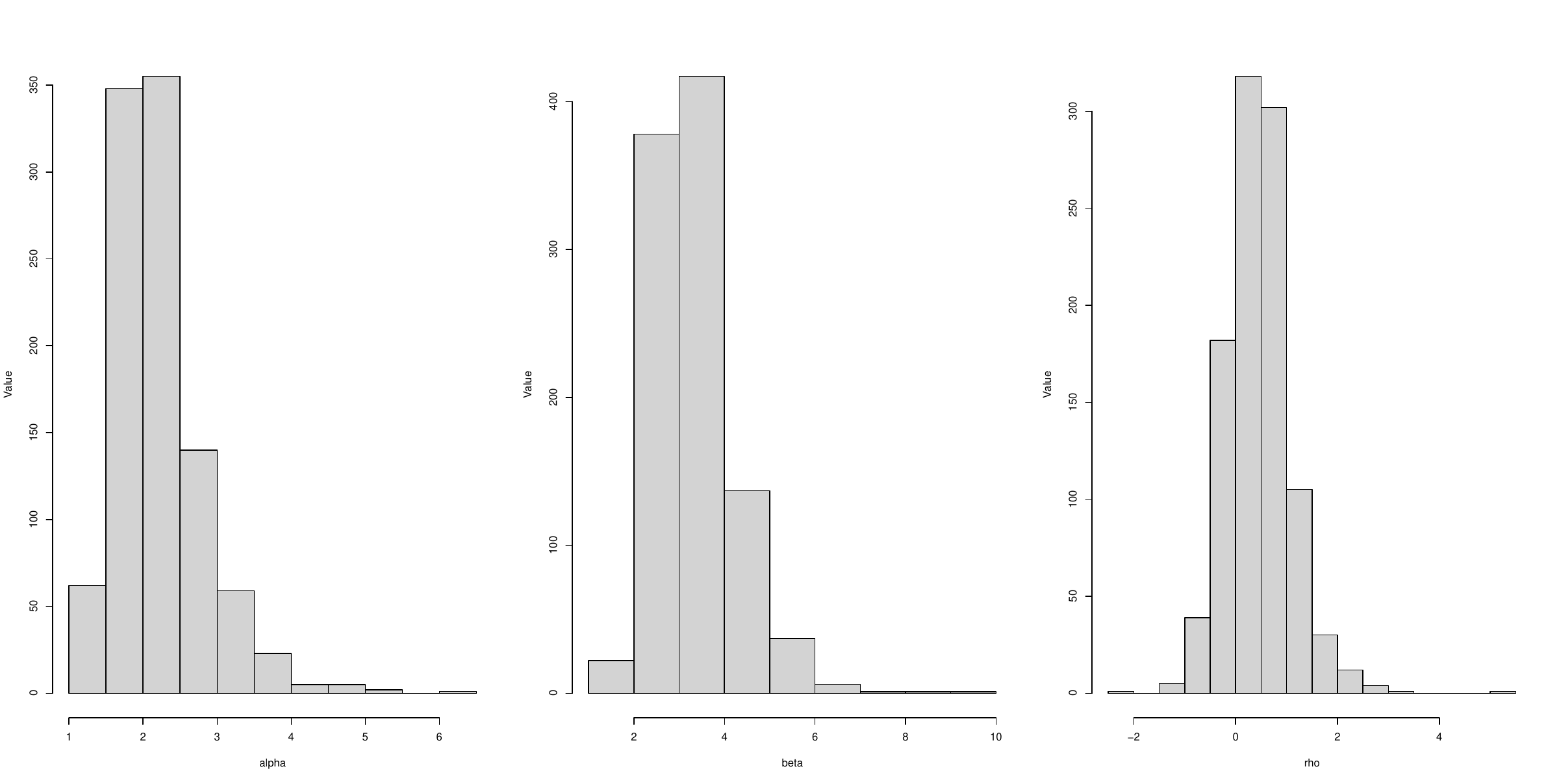}
            \includegraphics[width=0.9\linewidth, height=5cm]{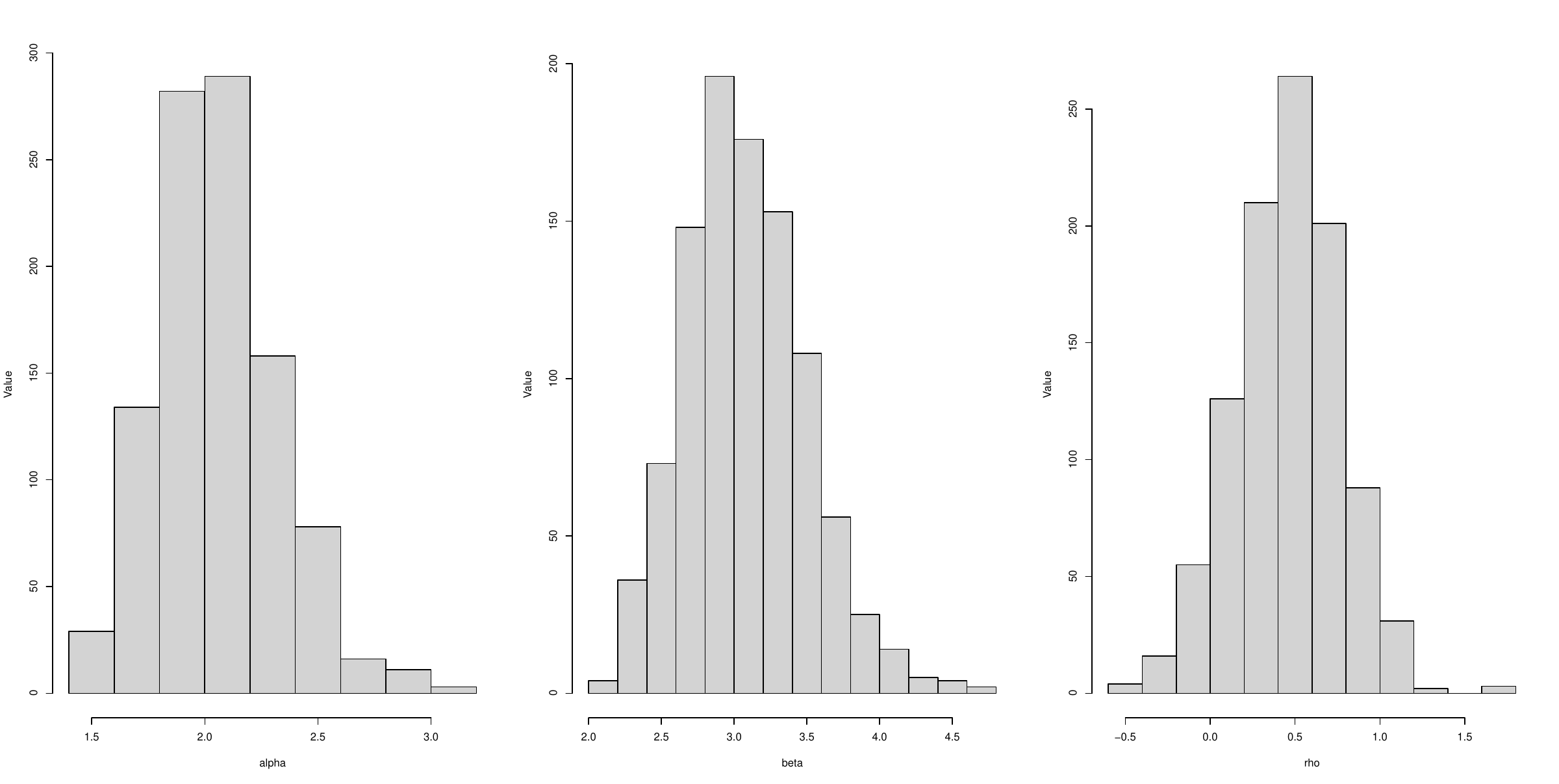}
            \includegraphics[width=0.9\linewidth, height=5cm]{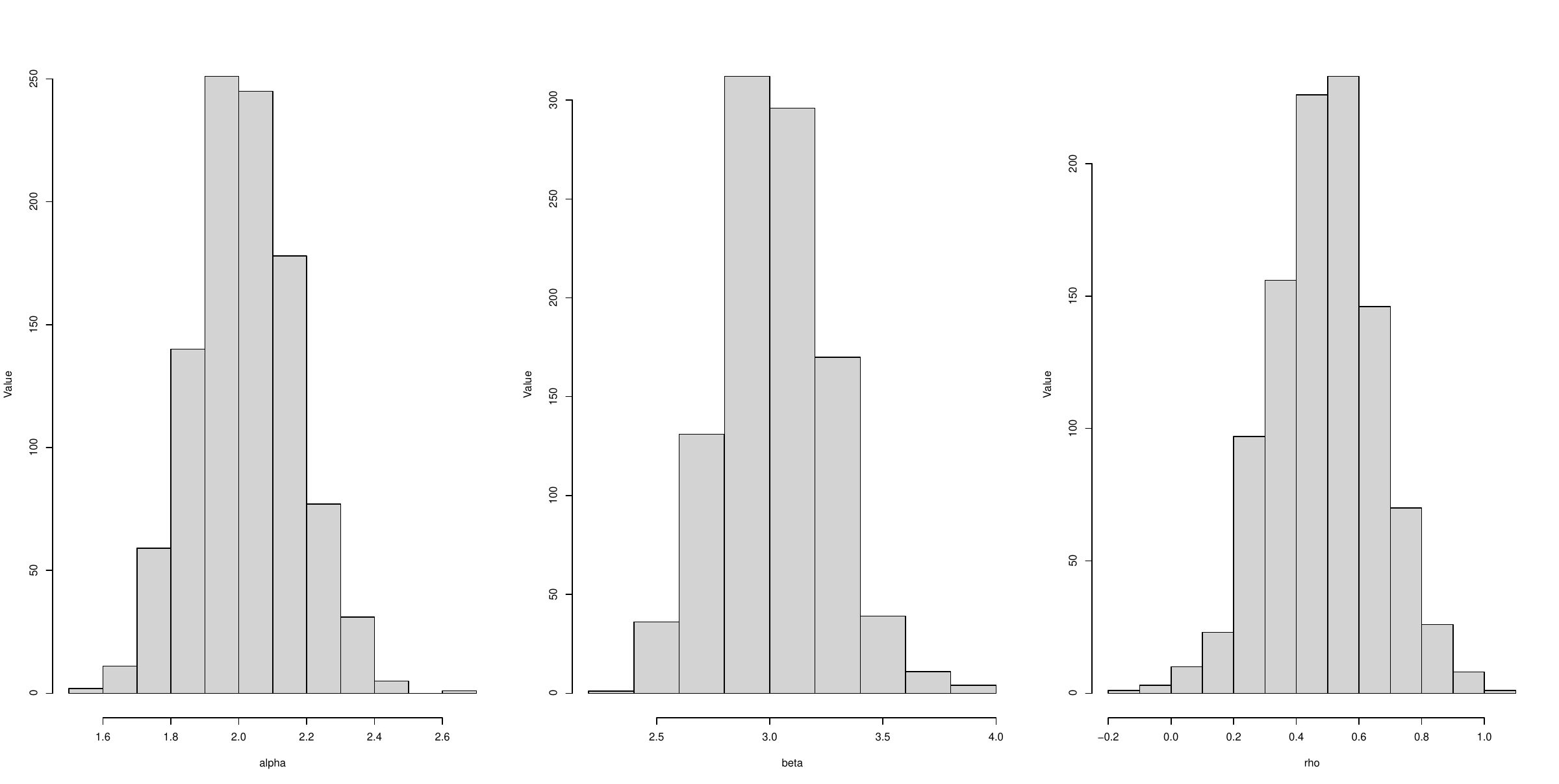}
            \includegraphics[width=0.9\linewidth, height=5cm]{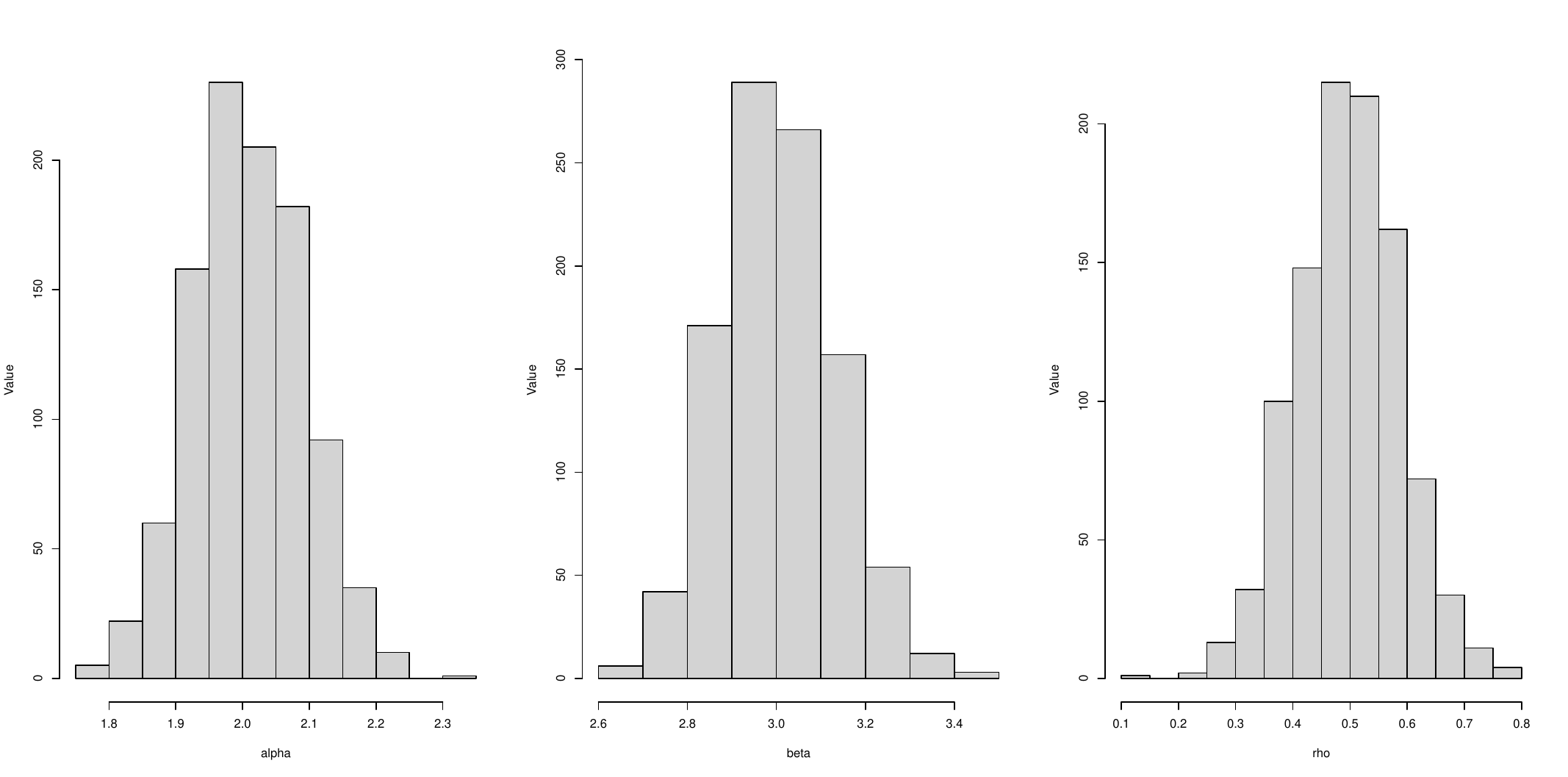}
	\caption{Histograms of Monte Carlo results for the {\color{black}PML} estimates of $\alpha$ (left), $\beta$ (middle), and $\rho$ (right), with increasing sample size $n = 30$ (top), 100, 300, and 1000 (bottom).}
	\label{fig:hist_MC_results}
\end{figure}

{\color{black}
Although gradient descent is a simple optimization algorithm, it may not be the most suitable choice for bounded-parameter likelihoods. Figure \ref{fig:algo_comparison} compares the estimation of $\rho$ using gradient descent and the \textsf{L-BFGS-B} method. The two approaches exhibit similar behavior, except in the extreme cases when $\rho = \pm 0.95$, when the \textsf{L-BFGS-B} performs better for large sample sizes.

\begin{figure}[H]
	\centering
            \includegraphics[width=1.0\linewidth, height=12cm]{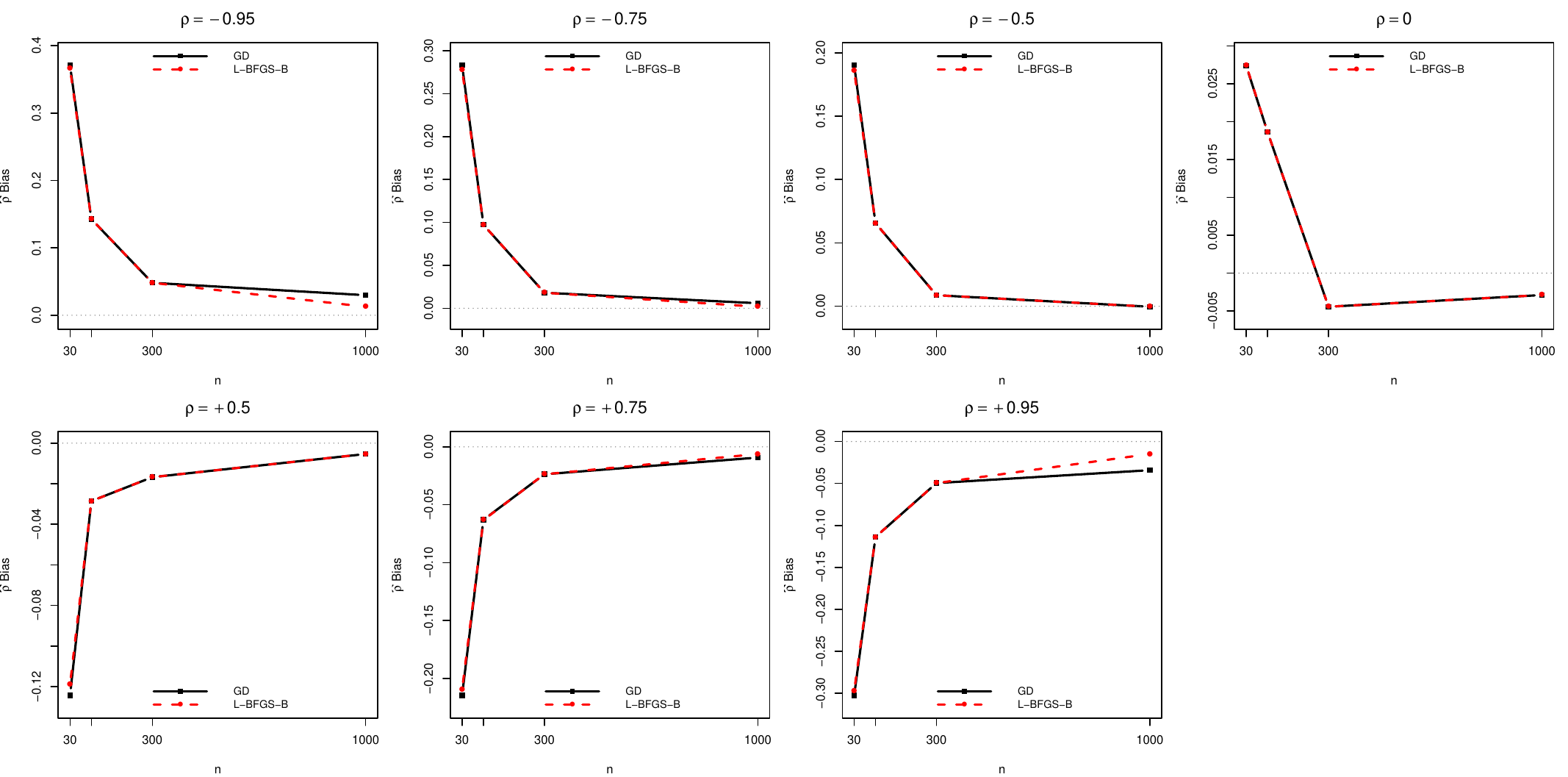}
	\caption{{\color{black}Monte Carlo results for the Bias decay of PML estimates of $\rho\in\{\pm0.95, \pm0.75, \pm0.5, 0\}$, with increasing sample size $n = 30, 100, 300, 1000$.}}
	\label{fig:algo_comparison}
\end{figure}
}

{\color{black}
\subsection{Identifiability through Fisher information}
\noindent

In the previous subsection, some parameter configurations, such as $\boldsymbol{\theta}=(2,3,0.5)^\top$, were associated with relatively large RMSE values for small sample sizes. Motivated by this behavior, we investigate the presence of weak non-identifiability in a neighborhood of $\boldsymbol{\theta}$.

To this end, we numerically evaluate the Fisher information matrix \eqref{eq:Fisher_information} for $\boldsymbol{\theta}_0=(2,3,\rho)^\top$, varying $\rho\in(-1,1)$. Specifically, we analyze the following indicators of weak non-identifiability:
the proximity of the smallest eigenvalue to zero ($\lambda_{\min}\approx 0$); the condition number ($\kappa\left(I(\boldsymbol{\theta}_0)\right) = \lambda_{max}/\lambda_{min}\uparrow\infty$); and the proximity of the determinant values to zero ($\det\left(I(\boldsymbol{\theta}_0)\right)\approx0$). 
Tables \ref{tab:fisher} and \ref{tab:fisher_zero} present the results for $|\rho|\approx 1$ and $|\rho|\approx 0$, respectively. As $|\rho| \uparrow 1$, $\lambda_{\min}$ appears to remain essentially
unchanged with respect to $\rho$, while
$\kappa\!\left(I(\boldsymbol{\theta}_0)\right)$ and
$\det\!\left(I(\boldsymbol{\theta}_0)\right)$ remain above the thresholds of $10^{-6}$ and below $10^{6}$, respectively. Therefore, these measures do not indicate non-identifiability. However, these measures exhibit different scales for $\rho \in \{-0.75, 0, 0.75\}$. Then, a more detailed investigation was conducted in the neighborhood of $|\rho| \approx 0$, which likewise revealed no indication of weak non-identifiability.

\begin{table}[H]
\centering
\caption{{\color{black} Determinant, condition number, and minimum eigenvalue of
         $I(\boldsymbol{\theta}_0)$ for $\boldsymbol{\theta}_0 = (2,3, \rho)$, varying $\rho\in\{\pm 0.99, \pm 0.95, \pm 0.90, \pm 0.75, 0 \}$.}}
\label{tab:fisher}
\begin{tabular}{rrrr}
\hline
$\rho$ & $\det\bigl(I(\boldsymbol{\theta}_0)\bigr)$ & $\kappa\bigl(I(\boldsymbol{\theta}_0)\bigr)$ &
         $\lambda_{\min}$ \\
\hline
$-0.99$ & $3.083 \times 10^{-4}$ &   51.16 & 0.00849 \\
$-0.95$ & $3.006 \times 10^{-4}$ &   53.03 & 0.00824 \\
$-0.90$ & $2.911 \times 10^{-4}$ &   55.54 & 0.00793 \\
$-0.75$ & $2.640 \times 10^{-4}$ &   64.44 & 0.00701 \\
$ 0.00$ & $1.398 \times 10^{-4}$ &  184.68 & 0.00295 \\
$+0.75$ & $4.066 \times 10^{-5}$ & 1261.85 & 0.00056 \\
$+0.90$ & $4.958 \times 10^{-5}$ & 1242.86 & 0.00060 \\
$+0.95$ & $6.078 \times 10^{-5}$ & 1083.94 & 0.00071 \\
$+0.99$ & $7.461 \times 10^{-5}$ &  934.03 & 0.00083 \\
\hline
\end{tabular}
\end{table}

\begin{table}[ht]
\centering
\caption{{\color{black} Determinant and condition number of
         $I(\boldsymbol{\theta}_0)$ for $\boldsymbol{\theta}_0 = (2,3, \rho)$, varying $\rho\in\{\pm 0.10, \pm 0.05, \pm 0.01, 0 \}$.}}
\label{tab:fisher_zero}
\begin{tabular}{rrr}
\hline
$\rho$ & $\det\bigl(I(\boldsymbol{\theta}_0)\bigr)$ & $\kappa\bigl(I(\boldsymbol{\theta}_0)\bigr)$ \\
\hline
$-0.10$ & $1.560 \times 10^{-4}$ & 154.19 \\
$-0.05$ & $1.479 \times 10^{-4}$ & 168.40 \\
$-0.01$ & $1.414 \times 10^{-4}$ & 181.24 \\
$ 0$ & $1.398 \times 10^{-4}$ & 184.68 \\
$+0.01$ & $1.382 \times 10^{-4}$ & 188.22 \\
$+0.05$ & $1.317 \times 10^{-4}$ & 203.44 \\
$+0.10$ & $1.236 \times 10^{-4}$ & 225.18 \\
\hline
\end{tabular}
\end{table}

}

{\color{black}
\subsection{Simulation results for the MLE}
\noindent

As highlighted in the previous subsections, the parameter region $\boldsymbol{\theta} = (2,3,\rho)$ with $\rho \in (-1,1)$ poses a considerable challenge, exhibiting high RMSE for small sample sizes. This configuration is therefore adopted for the subsequent analysis.

We now evaluate the MLE obtained by maximizing the log-likelihood function (\ref{loglik}). The optimization was performed using the ``L-BFGS-B'' method implemented in the \textsf{optim} function. For this purpose, we consider the 
following performance measures: Bias, RMSE, coverage probability (CP), computational time in seconds (CT), and the percentage of non-convergence (PNC).

Table \ref{tab:mle} reports the MLE results for $\rho \in \{0, \pm 0.75, \pm 0.5\}$, based on $M = 1,000$ Monte Carlo 
replications with sample sizes $n \in \{30, 50, 100\}$. 
Despite the high computational cost, the bias and RMSE remain relatively large, even for $n=100$, particularly for the estimates of $\rho$. The performance deteriorates further for large positive values of $\rho \in \{0.75, 0.95\}$. The PNC is low, indicating that the algorithms converge reliably, although the sample sizes are insufficient to ensure accurate estimation. The CP values are consistently high, suggesting that the Hessian-based confidence intervals are overly wide.

\begin{table}[!htbp]
\caption{ {\color{black} Monte Carlo results for the MLE of $\boldsymbol{\theta}=(2,3,\rho)$, $\rho\in\{0, \pm 0.75, \pm 0.95\}$, with increasing sample size $n\in\{30, 50, 100\}$.}}
\centering
\begin{tabular}{lcrcccc}
  \hline
Parameter & $n$ & Bias & RMSE & CP & CT  & PNC \\ 
  \hline
$\alpha=2$ & 30 & 0.05600 & 0.58956 & 0.9905 & 767.20 & 0.00 \\
 & 50  & $-$0.02656 & 0.43703 & 0.9857 & 1411.89 & 0.20 \\
  &100  & $-$0.06042 & 0.32979 & 0.9843 & 2351.14 & 0.30 \\ 
  $\beta=3$& 30 & 0.08018 & 0.88569 & 0.9893 & 767.20 & 0.00 \\
   &  50  & $-$0.05671 & 0.69583 & 0.9879 & 1411.89 & 0.20 \\ 
  & 100  & $-$0.10326 & 0.52213 & 0.9823 & 2351.14 & 0.30 \\ 
 $\rho=-0.95$& 30 & 0.33787 & 0.70863 & 0.9964 & 767.20 & 0.00 \\ 
 &  50 & 0.35338 & 0.71369 & 0.9890 & 1411.89 & 0.20 \\ 
  & 100 & 0.29689 & 0.62501 & 0.9864 & 2351.14 & 0.30 \\ 
   \hline
$\alpha=2$&30 & 0.11124 & 0.63604 & 0.9917 & 748.53 & 0.00 \\ 
&  50 & 0.02210 & 0.46759 & 0.9890 & 1255.88 & 0.20 \\ 
&    100 & $-$0.02313 & 0.35448 & 0.9864 & 2426.31 & 0.20 \\ 
 $\beta=3$&   30 & 0.16851 & 0.95772 & 0.9917 & 748.53 & 0.00 \\ 
&     50 & 0.02206 & 0.74557 & 0.9868 & 1255.88 & 0.20 \\ 
&       100 & $-$0.04233 & 0.56078 & 0.9832 & 2426.31 & 0.20 \\ 
$\rho=-0.75$&  30 & 0.21429 & 0.70100 & 0.9964 & 748.53 & 0.00 \\ 
 & 50  & 0.23314 & 0.70115 & 0.9901 & 1255.88 & 0.20 \\ 
  &100  & 0.19733 & 0.63401 & 0.9843 & 2426.31 & 0.20 \\ 
   \hline
$\alpha=2$&30 & 0.35182 & 0.79184 & 0.9964 & 854.45 & 0.20 \\ 
&  50 & 0.25137 & 0.64850 & 0.9946 & 1414.72 & 0.30 \\ 
&    100& 0.12106 & 0.47688 & 0.9903 & 2775.91 & 0.40 \\ 
 $\beta=3$& 30 & 0.56009 & 1.22529 & 0.9976 & 854.45 & 0.20 \\
&   50 & 0.39668 & 1.03399 & 0.9946 & 1414.72 & 0.30 \\ 
&    100& 0.19406 & 0.75783 & 0.9924 & 2775.91 & 0.40 \\ 
$\rho=0$&  30 & -0.24408 & 0.82535 & 1.0000 & 854.45 & 0.20 \\ 
  &50 & $-$0.20709 & 0.78711 & 1.0000 & 1414.72 & 0.30 \\ 
  &100 & $-$0.10590 & 0.72040 & 0.9968 & 2775.91 & 0.40 \\ 
   \hline
$\alpha=2$&30 & 0.51064 & 0.87196 & 1.0000 & 866.50 & 0.10 \\
&  50 & 0.43773 & 0.75617 & 1.0000 & 1459.34 & 0.20 \\ 
 & 100 & 0.32610 & 0.61528 & 0.9814 & 2966.69 & 0.20 \\ 
 $\beta=3$& 30 & 0.81821 & 1.36724 & 1.0000 & 866.50 & 0.10 \\
&   50 & 0.69083 & 1.17904 & 0.9986 & 1459.34 & 0.20 \\
&     100 & 0.52790 & 0.98317 & 0.9787 & 2966.69 & 0.20 \\ 
 $\rho=0.75$ &30 & -0.46352 & 0.92941 & 1.0000 & 866.50 & 0.10 \\ 
  &50 & $-$0.46818 & 0.90648 & 1.0000 & 1459.34 & 0.20 \\ 
  &100 & $-$0.40425 & 0.81429 & 0.9322 & 2966.69 & 0.20 \\ 
   \hline
$\alpha=2$&30 & 0.58795 & 0.94845 & 1.0000 & 895.58 & 0.10 \\ 
&  50 & 0.47393 & 0.78774 & 0.9984 & 1470.57 & 0.30 \\ 
&    100  & 0.36103 & 0.62723 & 0.9829 & 2868.99 & 0.50 \\ 
$\beta=3$&  30 & 0.94834 & 1.50294 & 1.0000 & 895.58 & 0.10 \\
&    50 & 0.75136 & 1.23516 & 0.9935 & 1470.57 & 0.30 \\ 
&      100 & 0.58462 & 1.00159 & 0.9720 & 2868.99 & 0.50 \\ 
$\rho=0.95$&  30  & -0.53292 & 0.92458 & 1.0000 & 895.58 & 0.10 \\ 
 & 50  & $-$0.48614 & 0.85692 & 0.9967 & 1470.57 & 0.30 \\ 
 & 100  & $-$0.42980 & 0.76656 & 0.9068 & 2868.99 & 0.50 \\ 
   \hline
\end{tabular}
\label{tab:mle}
\end{table}

Table \ref{tab:init_sensitivity} presents the sensitivity of the MLE to different starting values. For the true parameter $\boldsymbol{\theta} = (2,3,0)^\top$, $M = 1,000$ Monte Carlo replications were conducted using five fixed initializations: the true value $(2,3,0)^\top$, a value close to the true parameter $(2.5,3.5,0.2)^\top$, a 
standard parameter $(1,1,0)^\top$, a distant parameter $(0.5,0.5,0.8)^\top$, and an opposite direction $(5,5,-0.8)^\top$. The Bias and RMSE were not influenced by the choice of initial value.

\begin{table}[H]
\caption{{\color{black} Monte Carlo results for the MLE of $\boldsymbol{\theta}=(2,3,0)$, with fixed sample size $n= 100$, and varying the initial values.}}
\centering
\begin{tabular}{ccrccc}
  \hline
 Initial values & Parameter & Bias & RMSE & CP & PNC \\ 
  \hline
 $(2,3,0)^\top$ & $\alpha$ & 0.12106 & 0.47688 & 0.9903 & 0.40 \\ 
  & $\beta$ & 0.19406 & 0.75783 & 0.9924 & 0.40 \\ 
   & $\rho$ & $-$0.10590 & 0.72040 & 0.9968 & 0.40 \\ 
   \hline
  $(2.5,3.5,0.2)^\top$ & $\alpha$ & 0.12386 & 0.48056 & 0.9892 & 0.30 \\ 
   & $\beta$ & 0.19880 & 0.76614 & 0.9914 & 0.30 \\ 
   & $\rho$ & $-$0.10771 & 0.71892 & 0.9968 & 0.30 \\ 
   \hline
   $(1,1,0)^\top$ & $\alpha$ & 0.11381 & 0.47975 & 0.9902 & 0.40 \\ 
   & $\beta$ & 0.18071 & 0.76330 & 0.9924 & 0.40 \\ 
   & $\rho$ & $-$0.09876 & 0.72179 & 0.9967 & 0.40 \\ 
   \hline
   $(0.5,0.5,0.8)^\top$ & $\alpha$ & -0.09965 & 0.43539 & 0.9884 & 0.40 \\ 
   & $\beta$ & $-$0.17015 & 0.68812 & 0.9913 & 0.40 \\ 
   & $\rho$ & 0.29312 & 0.77951 & 0.9956 & 0.40 \\ 
  \hline
  $(5,5,-0.8)^\top$ & $\alpha$ & 0.13310 & 0.50888 & 0.9882 & 0.30 \\ 
   & $\beta$ & 0.20552 & 0.77643 & 0.9903 & 0.30 \\ 
   & $\rho$ & $-$0.11266 & 0.72013 & 0.9968 & 0.30 \\ 
   \hline
\end{tabular}
\label{tab:init_sensitivity}
\end{table}
}

{\color{black}
\subsection{Comparison of estimators}
\noindent

Consider a random sample $(\mathbf{X}, \mathbf{Y}) = ((X_1,Y_1)^\top, \ldots, (X_n, Y_n)^\top)^\top$, where $(X,Y)$  follows a Morgenstern’s bivariate distribution with gamma margins \eqref{Morgenstern-cdf}. Then, we can generate the sample $\mathbf{Z}=(Z_1,\cdots,Z_n)\sim\operatorname{EB}(\boldsymbol{\theta})$ from the ratio \eqref{rep-stoch-1-1}. In this subsection, we compare the MLE results in Table \ref{tab:mle} with the two-stage PMLE, using the Algorithm \ref{algo:gamma_estimation}.

Table \ref{tab:twostep_rho} shows that the PMLE (Algorithm~\ref{algo:gamma_estimation}) exhibits behavior similar to the MLE in the estimation of $\alpha$ and $\beta$. However, the PMLE was more accurate in the estimation of $\rho$, while requiring considerably less 
computational time. On the other hand, the coverage probability for $\rho$ performs worse when compared with the MLE results. In Section~\ref{sec:applications}, a real-data application is presented for a large sample size. Due to the computational cost, the PMLE emerges as the preferable approach.

\begin{table}[H]
\caption{ {\color{black} Monte Carlo results for the PMLE of $\boldsymbol{\theta}=(2,3,\rho)$, $\rho\in\{0, \pm 0.75, \pm 0.95\}$, with increasing sample size $n\in\{30, 50, 100\}$.}}
\centering
\begin{tabular}{lcccccc}
  \hline
Parameter & $n$ & Bias & RMSE & CP & CT  & PNC \\ 
  \hline
$\alpha=2$ & 30 & 0.20498 & 0.61157 & 0.9569 & 1.34 & 0.20 \\  & 50 & 0.12671 & 0.42604 & 0.9639 & 1.80 & 0.30 \\ 
&100 & 0.06193 & 0.27979 & 0.9567 & 3.13 & 0.60 \\

$\beta=3$ &  30 & 0.29971 & 0.94568 & 0.9619 & 1.34 & 0.20 \\  & 50 & 0.16644 & 0.64704 & 0.9619 & 1.80 & 0.30 \\ 
& 100 & 0.07672 & 0.43218 & 0.9437 & 3.13 & 0.60 \\

 $\rho=-0.95$& 30 & 0.36592 & 0.64695 & 0.4749 & 1.34 & 0.20 \\ 
  &50 & 0.24496 & 0.45351 & 0.5416 & 1.80 & 0.30 \\ 
  &100  & 0.14261 & 0.29062 & 0.5775 & 3.13 & 0.60 \\ 
   \hline
$\alpha=2$ & 30 & 0.20647 & 0.61553 & 0.9529 & 1.39 & 0.20 \\ & 50 & 0.12889 & 0.43260 & 0.9599 & 1.92 & 0.30 \\  
& 100 & 0.06171 & 0.28041 & 0.9567 & 3.48 & 0.60 \\

$\beta=3$ &   30& 0.29971 & 0.94568 & 0.9619 & 1.39 & 0.20 \\  & 50 & 0.16644 & 0.64704 & 0.9619 & 1.92 & 0.30 \\   & 100 & 0.07672 & 0.43218 & 0.9437 & 3.48 & 0.60 \\

 $\rho=-0.75$&  30  & 0.27773 & 0.61676 & 0.6202 & 1.39 & 0.20 \\ 
&   50  & 0.17840 & 0.45499 & 0.6901 & 1.92 & 0.30 \\ 
&   100 & 0.09727 & 0.30915 & 0.8048 & 3.48 & 0.60 \\ 
   \hline
$\alpha=2$ & 30 &  0.21488 & 0.63130 & 0.9549 & 1.39 & 0.30 \\  & 50 &  0.13538 & 0.45079 & 0.9548 & 2.27 & 0.40 \\  &100 &  0.06143 & 0.28454 & 0.9538 & 3.67 & 0.50 \\

 $\beta=3$ &  30 & 0.30062 & 0.94596 & 0.9619 & 1.39 & 0.30 \\  & 50  & 0.16670 & 0.64736 & 0.9618 & 2.27 & 0.40 \\  &100  & 0.07704 & 0.43204 & 0.9437 & 3.67 & 0.50 \\

 $\rho=0$& 30  & 0.02746 & 0.55469 & 0.8195 & 1.39 & 0.30 \\ 
 & 50  & 0.02814 & 0.44205 & 0.8886 & 2.27 & 0.40 \\ 
 & 100  & 0.01863 & 0.31960 & 0.9156 & 3.67 & 0.50 \\ 
   \hline
$\alpha=2$ & 30 & 0.22374 & 0.64224 & 0.9599 & 1.50 & 0.20 \\   &50 & 0.13919 & 0.45725 & 0.9569 & 2.08 & 0.30 \\ 
  &100  & 0.06213 & 0.28982 & 0.9467 & 3.75 & 0.50 \\ 

 $\beta=3$ &  30  & 0.30077 & 0.94560 & 0.9619 & 1.50 & 0.20 \\  
 &50  & 0.16744 & 0.64766 & 0.9619 & 2.08 & 0.30 \\   &100  & 0.07691 & 0.43177 & 0.9437 & 3.75 & 0.50 \\

$\rho=0.75$&  30   & -0.20944 & 0.50820 & 0.6293 & 1.50 & 0.20 \\ 
&  50  & -0.12734 & 0.37565 & 0.7141 & 2.08 & 0.30 \\ 
 & 100  & -0.06316 & 0.25994 & 0.8040 & 3.75 & 0.50 \\ 
   \hline
$\alpha=2$ &30 & 0.22562 & 0.64469 & 0.9580 & 1.70 & 0.10 \\ 
 & 50 & 0.13936 & 0.45787 & 0.9559 & 2.14 & 0.20 \\  
 &100 & 0.06227 & 0.29133 & 0.9448 & 3.96 & 0.30 \\ 
 
 $\beta=3$ &  30  & 0.30171 & 0.94593 & 0.9620 & 1.70 & 0.10 \\ 
 &50  & 0.16743 & 0.64736 & 0.9619 & 2.14 & 0.20 \\ 
  & 100  & 0.07771 & 0.43243 & 0.9438 & 3.96 & 0.30 \\ 

 $\rho=0.95$ & 30  & -0.29735 & 0.51256 & 0.5035 & 1.70 & 0.10 \\ 
 & 50 & -0.19327 & 0.36092 & 0.5561 & 2.14 & 0.20 \\ 
  &100  & -0.11379 & 0.22801 & 0.5807 & 3.96 & 0.30 \\ 
   \hline
\end{tabular}

\label{tab:twostep_rho}
\end{table}

}

{\color{black}
\subsection{Computational issues}
\noindent

We conclude this section by highlighting some computational challenges related to the implementation and validation of both the model and its estimators.

\begin{itemize}
\item Series truncation and evaluation accuracy: The double series \eqref{eq:F2}, which defines the Appell $F_2$ function for $|x| + |y| < 1$, can be computationally demanding. For our purposes, we restrict the double summation to a maximum of 200 terms per index and adopt a tolerance parameter of $\varepsilon = 10^{-12}$. To ensure accuracy, we compared selected configurations with the \textsf{mpmath.appellf2} function in Python, obtaining 
differences below $10^{-6}$ in all tests. In addition, some well-known identities were verified, yielding satisfactory results.

\item Numerical issues encountered in practice: The PDF and CDF of the EB model can be computationally intensive when applied to large datasets or in likelihood-based estimation procedures. Pre-computed lookup tables for special functions can provide an efficient alternative for practical implementations.

\item Extreme parameter values: When $|\rho| \uparrow 1$, the estimation yields poorer results. The model may become inapplicable when the data $(X,Y)$ are strongly associated, leading to the generation of a ratio random variable $Z$.

\end{itemize}
These computational challenges and limitations do not compromise the consistent results established for the PMLE and the EB model.
}

\section{Real-world data analysis}\label{sec:applications}
\noindent

In this section, we present an application of the proposed model to real data for illustrative purposes. To ensure replicability of our results and provide support for future users of the EB model, all R codes are publicly available at \url{https://github.com/FSQuintino/extended_beta_model.git} (accessed on 07 Jan 2026). {\color{black}A comparison with the well-known Beta and Kumaraswamy models is presented. We also consider a broad range of models with one, two, three, and four parameters, including bimodal distributions that have become well established in the recent literature. For the reader's convenience, in Appendix \ref{ap:competitor_models} we present the PDFs of the competing models.} 




\subsection{{\color{black} Income-consumption data}}
\noindent

The dataset comes from the Bank of Italy's Survey on Household Income and Wealth from the year 2008. 
We utilized data from the RFAM08 data set (with income as the $Y$ variable) and the RISFAM08 data set (with expenditures as the $C$ variable), as outlined in the Survey on Household Income and Wealth 2008 data description file. We excluded any data entries where income or consumption was negative or unavailable. The data were recently reported in \cite{vila2025novel} and the references therein, and are available at \url{https://www.bancaditalia.it/statistiche/tematiche/indagini-famiglie-imprese/bilanci-famiglie/documentazione/ricerca/ricerca.html?min_anno_pubblicazione=2008&max_anno_pubblicazione=2008} (accessed on 20 November 2025).

In this study, we analyze {\color{black} income-consumption data} by constructing a new variable $Z=Y/(X+Y)$ based on \eqref{eq:ratio_model}, where $X$ and $Y$ represent expenditures and income. Values of $Z < 1/2$ indicate that the family ended the year with more expenditures than income. If $Z > 1/2$, the opposite is true. The case $Z = 1/2$ means expenditures and income are equal.

{\color{black} We are interested in modeling $Z$
 because it is a one-dimensional bounded variable, easy to interpret, and it summarizes the income–consumption profile of families in a single value.
The expected behavior is a unit distribution with left skewness. A full characterization of $Z$
 allows us to describe families more effectively and offers an alternative way to study copula-based stress–strength probabilities $\mathbb{P}(X<Y)$, avoiding the limitations of Morgenstern's copula in cases of weak dependence. We then show that the EB model fits $Z$ appropriately.}


\subsection{{\color{black} Evaluating the EB for the ratio data}}
\noindent

In Table \ref{tab:datasets_descritive}, we present the following descriptive statistics: the sample size ($n$); the minimum value observed in the dataset (Min.);
the first quartile value (1st Qu.); the median value (Median); the average value (Mean); the third quartile value (3rd Qu.); the maximum value observed (Max.); the standard deviation (Sd.); the coefficient of symmetry (CS); the coefficient of kurtosis (CK).

\begin{table}[H]
\centering
\caption{{ Descriptive statistics for the {\color{black} income-consumption} dataset.}}
\label{tab:datasets_descritive}
\begin{tabular}{cccccccccc}
\hline
  $n$  &   Min.  & 1st Qu. & Median & Mean       & 3rd Qu. & Max.   & Sd         & CS          & CK \\    
\hline
   7957 & 0.00 & 0.51 & 0.55 & 0.56 & 0.61 & 0.94 & 0.09 & $-$0.42 & 3.34 \\ 
\hline
\end{tabular}
\end{table}


To estimate the EB parameters for dataset $Z$, we applied \eqref{eq:rho_est2}.
%
%
{\color{black}Consider the ratio variable $$Z=\frac{Y}{Y+X}.$$}
%
Table \ref{tab:models_fit} gives the estimates, the Akaike information criterion (AIC), and the Bayesian information criterion (BIC) for the fitted models. 
Since the AIC and BIC values are smaller for the EB model than for the traditional Beta and Kumaraswamy models, this new model seems to be a competitive fit for these data. 
{\color{black} Among the recently proposed unit distributions reported in Table \ref{tab:models_fit}, only the Unit-Inverse Gaussian (UIG) and Unit-Fréchet (UF) distributions exhibit information criterion values lower than those of the EB model. The Bimodal Beta (BBeta) distribution also provides a competitive fit, with information criterion values close to those obtained for the EB model. However, the histogram and empirical CDF (ECDF) plots, together with the fitted PDFs and CDFs (respectively, Figures \ref{fig:hist_models} and \ref{fig:ecdf_models}), indicate that the UIG's performance is inferior to that of the EB, UF, and BBeta models.
}


\begin{table}[H]
\centering
\caption{Parameter estimates and model selection (AIC and BIC) for the {\color{black} income-consumption}  dataset.}
\label{tab:models_fit}
\begin{tabular}{ccccc}
  \hline
Model & $\hat{\boldsymbol{\theta}}^\top$  & $\ell(\hat{\bm \theta})$ & AIC & BIC \\ 
  \hline 
   EB & (16.10, 12.84, 0.78) & 7922.92 & { $-$15839.84} & { $-$15818.90} \\ 
  {\color{black} Topp-Leone} & 4.26 & 7204.62 & $-$14407.25 & $-$14400.27 \\ 
  {\color{black} Unit-Lindley} & 1.09 & 3620.73 & $-$7239.47 & $-$7232.49 \\ 
   Beta & (16.10, 12.84)   & 7832.86 & $-$15661.72 & $-$15647.76 \\ 
 Kumaraswamy & (6.46, 28.86)   & 7795.81 & $-$15587.63 & $-$15573.66 \\ 
      {\color{black}Logit-Normal} & (0.37, 0.87) & 4340.08 & $-$8676.16 & $-$8662.19 \\ 
  {\color{black}Unit-Gompertz} & (1.81, 0.75) & 2069.28 & $-$4134.56 & $-$4120.60 \\ 
  {\color{black}Unit-Inverse Gaussian} & (1.37, 5.10) & 12011.25 & $-$24018.49 & $-$24004.53 \\
 {\color{black} Unit-Weibull} & (2.63, 2.48) & 5737.09 & $-$11470.19 & $-$11456.22 \\
 {\color{black}Unit-Frechet} & (1.25, 5.52, 0.1949) & 8381.26 & $-$16756.52 & $-$16735.57 \\ 
{\color{black}UBBS} & (7.15, 0.76, 18.60) & 4673.84 & $-$9341.68 & $-$9311.76 \\     {\color{black}Bimodal Beta} & (14.86, 13.27, 15.23, 17.46) & 7874.68 & $-$15741.36 & $-$15713.44 \\ 
  \hline
\end{tabular}
\end{table}

{\color{black} \begin{remark}
The parameter estimates of the EB model reported in Table \ref{tab:models_fit} were obtained using the two-stage pseudolikelihood approach (Algorithm \ref{algo:beta_estimation}), requiring only 0.8467 minutes of computation time. Although MLE can also be employed, its computational cost becomes substantial for large datasets. In the present application, MLE required 413.8503 minutes to report results. While MLE increased the log-likelihood from $\ell(\widehat{\boldsymbol{\theta}}^{PMLE})=7922.9$ to $\ell(\widehat{\boldsymbol{\theta}}^{MLE})=195973.8$, the ``L-BFGS-B'' algorithm failed to converge, returning parameter estimates located on the boundary of the search space.
\end{remark} }

\begin{figure}[H]
	\centering
    	\includegraphics[width=0.9\linewidth, height=9cm]{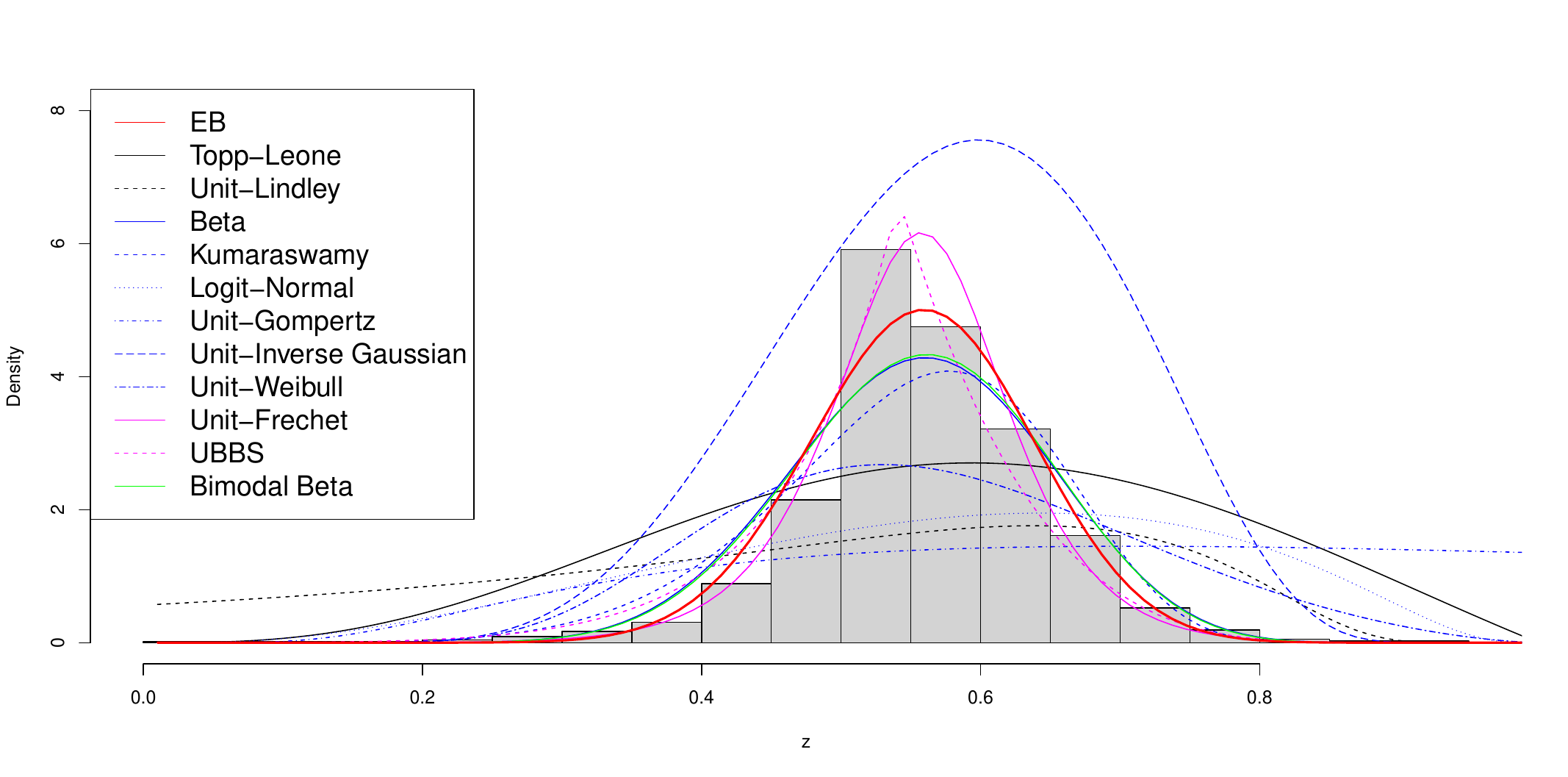}
	\caption{{\color{black} Histogram with the EB density (red) and competing one-parameter (black), two-parameter (blue), three-parameter (magenta), and four-parameter (green) models fitted to the income-consumption dataset.}}
	\label{fig:hist_models}
\end{figure}
\begin{figure}[H]
	\centering
	\includegraphics[width=1.0\linewidth, height=12cm]{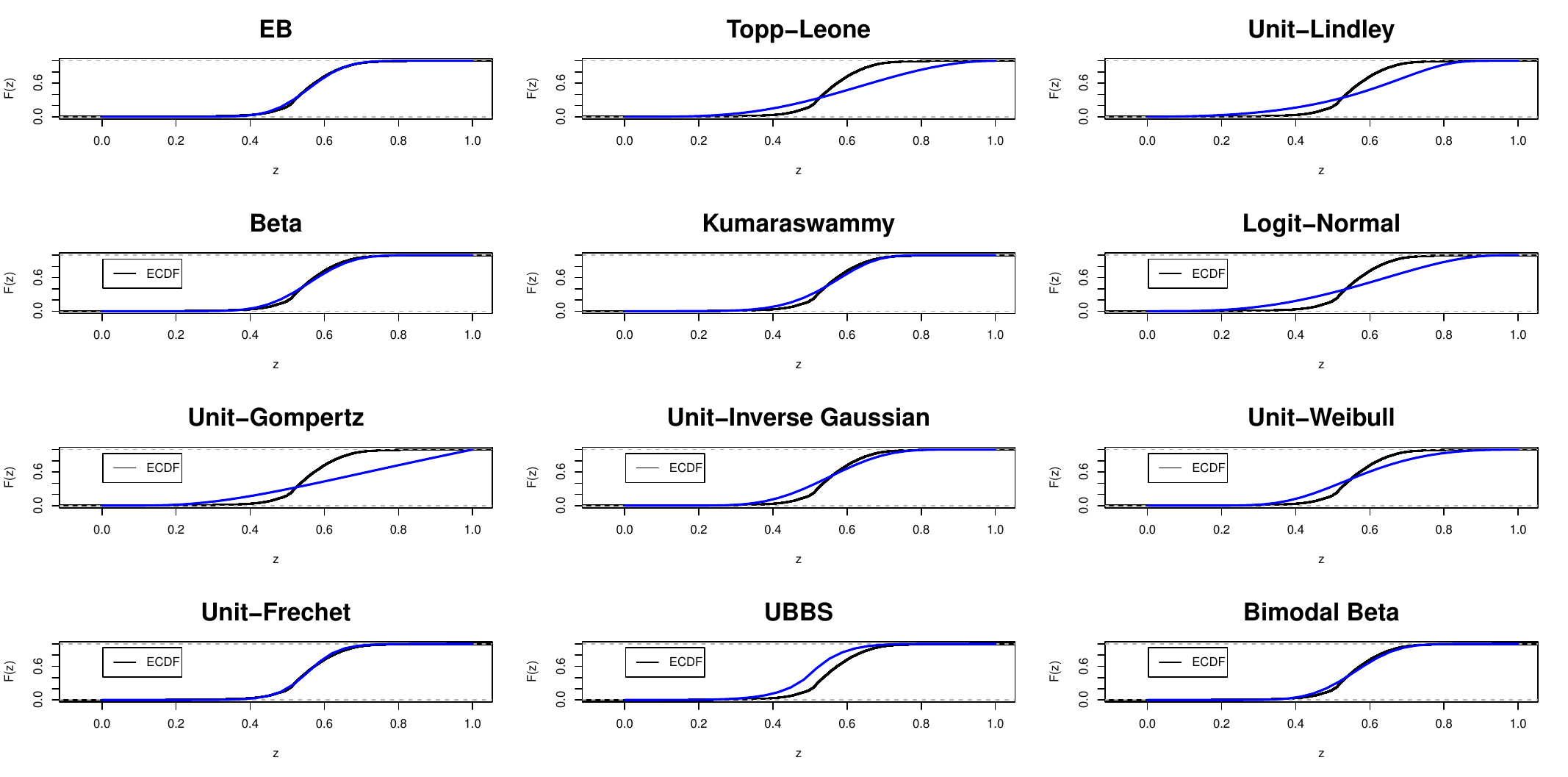}
	\caption{{\color{black}Empirical CDF (ECDF) compared with EB and competing models for the income-consumption dataset.}}
	\label{fig:ecdf_models}
\end{figure}

\begin{figure}[H]
	    \centering
	    \includegraphics[width=1.0\linewidth]{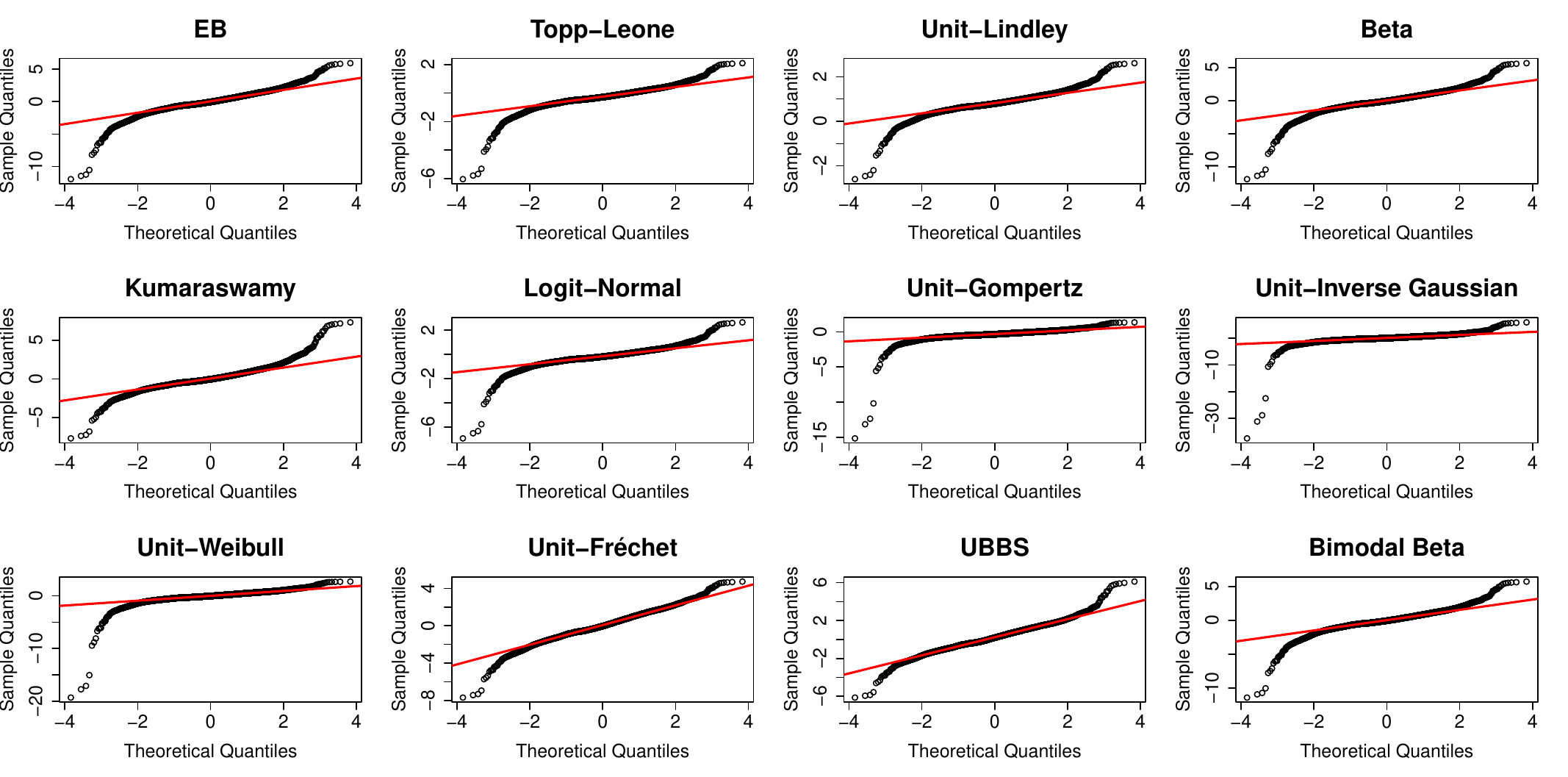}
	    \caption{Quantile-Quantile plot displaying randomized quantile residuals from fitted models for income-consumption data.}
	    \label{fig:qqplot}
	\end{figure}
\begin{figure}[H]
	    \centering
	    \includegraphics[width=1.0\linewidth]{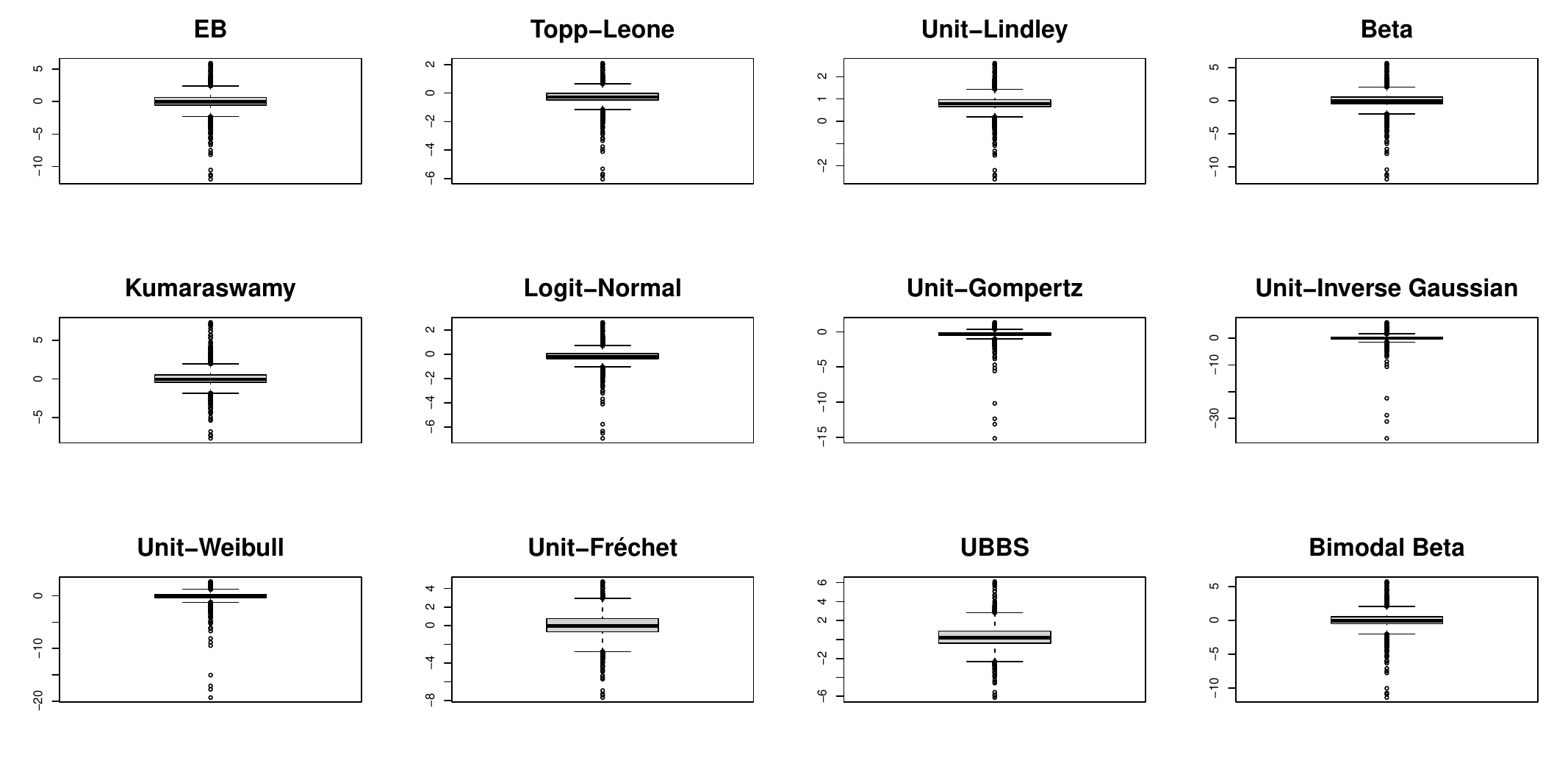}
	    \caption{Boxplot displaying randomized quantile residuals from fitted models for income-consumption data.}
	    \label{fig:residuals_boxplot}
	\end{figure}

{\color{black} To ensure a good model fit, we use randomized quantile residuals, defined by 
$R_i=\Phi^{-1}(F(z_i; \hat{\bm \theta}))$, where $\hat{\bm \theta}$ is the estimated parameter, $F(z_i; \hat{\bm \theta})$ is the CDF of the model fitted to each observation \(z_i\), and $\Phi^{-1}$ represents the quantile of the standard normal distribution $N(0, 1)$. 
Comparing the previously selected models (EB, UF, and EB), Figures \ref{fig:qqplot} and \ref{fig:residuals_boxplot} show the randomized quantiles are quite similar, with the UF model providing a slightly better fit in the left tail.

Given the large sample size ($n=7,957$), formal goodness-of-fit tests are expected to be highly sensitive and may strongly reject the null hypothesis even for models that provide an adequate fit in practical terms. Therefore, greater emphasis was placed on graphical diagnostics and information criteria when evaluating model performance, which showed that the EB model can be an interesting model for that data.
}

To test the null hypothesis $H_0: \rho = 0$ ($H_0$: beta model) against the alternative hypothesis $H_1: \rho \neq 0$ ($H_1$: EB model), we use the LR statistic, the value of which is 180.12 (p-value $<$ 0.01). Thus, using any usual significance level, we reject the null hypothesis in favor of the EB model, i.e., this model is significantly better than the beta model for explaining the current data.

\subsection{{\color{black} Assessing copula selection in income–consumption modeling}}
 \noindent

As pointed out in Subsection \ref{EB model arising as a ratio}, the {\color{black} EB} model arises from a ratio, where the random vector $(X, Y)^\top$ follows the Morgenstern bivariate distribution with gamma margins. 
Figure \ref{fig:gamma_margins} shows the gamma fits of the marginal distributions of expenditures and income. These findings offer more support for the {\color{black} EB} model's good fit to the derived $Z$ unitary data, {\color{black} which was discussed in detail in the previous subsection.
To estimate the Gamma parameters of the margins, we use the \textsf{fitdistr} function from the R package \textsf{MASS}. Due to the scale of the sample, the estimation procedure exhibits numerical instabilities.
Then, the adjustment $(X\cdot 10^{-3}, Y\cdot 10^{-3})$ is required to achieve numerical convergence in the estimation procedures. The choice of $10^3$ is based on trials with $10^k$. 
By varying $k$, the estimation for $X$ becomes numerically tractable at $k=1$, whereas $Y$ requires $k=2$. However, the estimated scale parameters remain too close to zero, leading us to choose $k=3$.
}

\begin{figure}[H]
	\centering
        \includegraphics[width=0.9\linewidth, height=7cm]{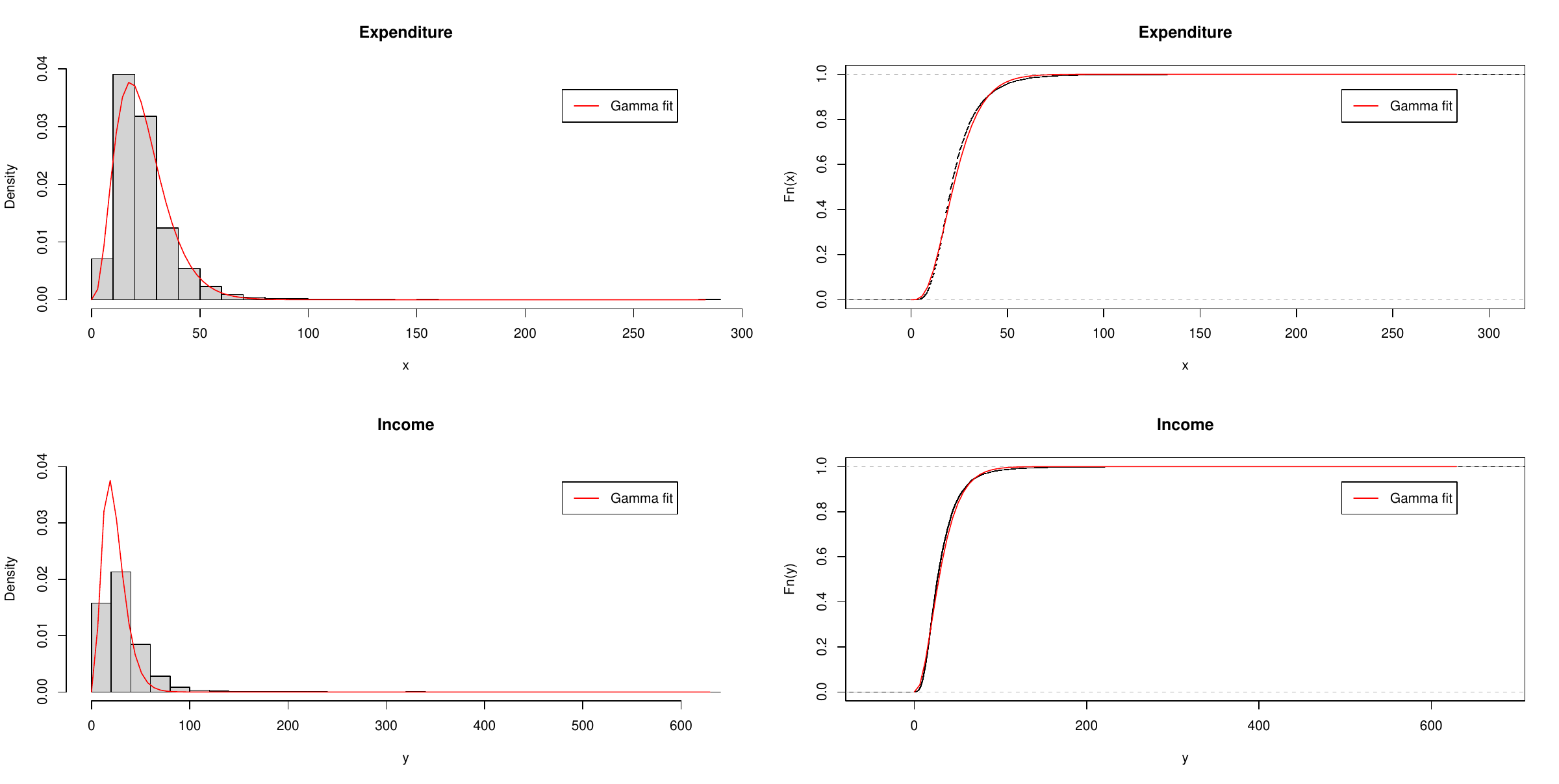}
       	\caption{Gamma fits to marginal distributions of expenditures and income.}
	\label{fig:gamma_margins}
\end{figure}

\begin{figure}[H]
	\centering
        \includegraphics[width=0.9\linewidth, height=9cm]{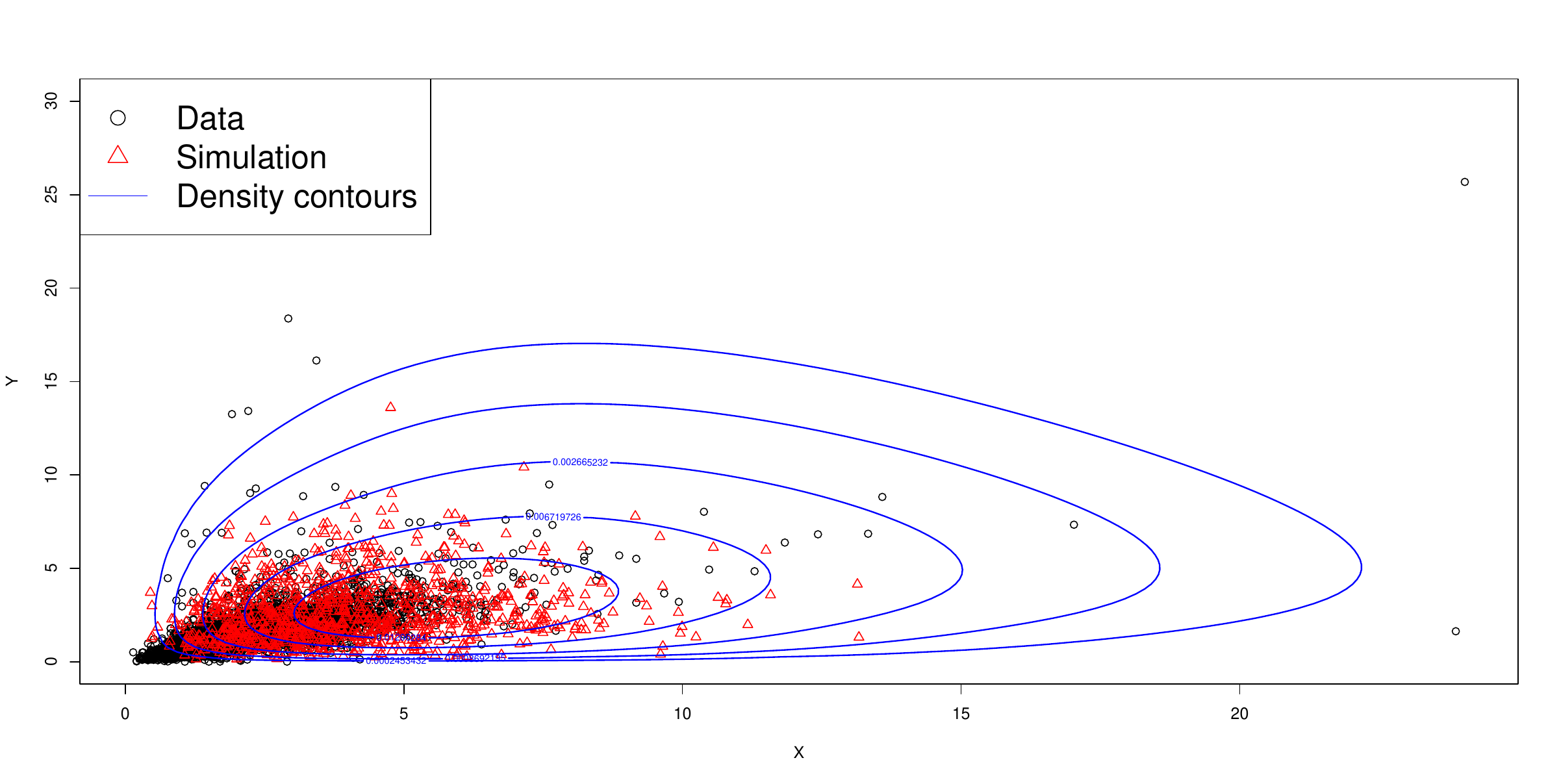}
       	\caption{Simulated samples of $C(Gamma(4.05, 2), Gamma(2.65,2); 0.77) $ (red) overlaid on the observed income–consumption data (black), along with the corresponding density contours (blue lines).}
	\label{fig:simulation_countor}
\end{figure}

{\color{black} After fitting the margins and aiming to use the estimator \eqref{eq:rho_est2} of $\rho$, we require that the data $(X,Y)$ share the same scale parameter $\theta$. We then transformed the data as follows:
$$x \leftarrow 0.5 \hat{\alpha} x /1000~~\text{and}~~y \leftarrow 0.5 \hat{\beta} y /1000.$$
Figure \ref{fig:simulation_countor} presents simulated samples overlaid on the observed data, along with the corresponding density contours. The model captures the joint behaviour of the data well, although a few outliers can be observed.

Table \ref{tab:copulas_selection} presents a comparative study involving some more flexible copulas from the extended FGM literature. For the reader's convenience, in Appendix \ref{ap:competitor_models} we present the CDFs of the benchmark copulas, named the Generalized
Farlie-Gumbel-Morgenstern (GFGM) copula.
A more detailed revision of 
GFGM copula, the reader can find information in \cite{amblard2002symmetry, susam2022multi}.
In Table \ref{tab:copulas_selection}, the log-likelihood values reveal no substantial differences among the GFGM generators. However, two of them yielded $\hat{\rho}=1$, suggesting that the model may not be suitable for capturing the dependence structure. Alternative copula families will be considered in future work. 
}

\begin{table}[H]
\caption{Copula dependence estimates, model selection based on the maximum log-likelihood value, and Spearman ($\rho_S$) and Kendall ($\tau$) corresponding estimates.}
\label{tab:copulas_selection}
\centering
\begin{tabular}{ccccc}
    \hline
  GFGM generator & $\widehat{\boldsymbol{\theta}}$ & $\ell(\widehat{\boldsymbol{\theta}})$ & $\rho_S(\hat{\rho})$ & $\tau(\hat{\rho})$ \\
  \hline
  $\phi(x) = min(x, 1-x)$ &  (4.05, 2.65, 0.59) & $-$51830.61 &0.20&0.13\\
    $\phi(x) =x(1-x)$ & (4.05, 2.65,0.77) & $-$51532.43 &0.26&0.17\\  
  $\phi(x) = x(1-x)(1-2x)$ &  (4.05, 2.65,1.00) & $-$52361.27 &0.33&0.22\\ 
  $\phi(x) = \frac{1}{\pi}\sin(\pi x)$ &  (4.05, 2.65,1.00) & $-$51054.64 &0.33&0.22\\
   \hline
\end{tabular}
\end{table}

\section{Concluding remarks}
\noindent

In this paper, we introduced the EB distribution, a new unit-supported model obtained from the ratio of two correlated Gamma random variables following a Morgenstern-type dependence structure. The additional dependence parameter extends the classical Beta distribution while preserving interpretability. 
{\color{black}
As shown throughout the paper, the EB family exhibits considerable flexibility in terms of distributional shape, encompassing symmetric and asymmetric densities, as well as boundary-inflated and other nonstandard behaviors that may arise in practice.
}
Closed-form expressions were obtained for the probability density and cumulative distribution functions and moments, relying on tools from special-function theory. Estimation was addressed using likelihood-based and alternative methods, and a Monte Carlo study highlighted both the finite-sample behavior of the estimators and the computational challenges that may arise. Applications to real datasets demonstrated that the EB distribution can provide improved fits compared to competing models, despite requiring only one additional parameter. 
Several directions remain open. Extensions to regression frameworks for responses to $(0,1)$, Bayesian inference with informative priors, and multivariate generalizations are natural avenues for future research. Another promising line involves studying robustness properties and developing faster numerical strategies for likelihood maximization in large samples.
{\color{black} New unit models also remain an alternative for future works, since the data like COVID-19 require copulas capable of capturing stronger associations than Morgenstern's copula, such as the Gumbel–Hougaard, Frank, and Clayton copulas with Gamma margins.}
In general, the EB model adds a versatile and analytically tractable option to the family of unit distributions, with potential for broad application in economics, reliability, hydrology, and related fields.

%
%
	\paragraph*{Acknowledgements}
The research was supported in part by CNPq and CAPES grants from the Brazilian government.
	
	\paragraph*{Disclosure statement}
	There are no conflicts of interest to disclose.


\bibliographystyle{apalike}
\bibliography{sample}

\begin{appendices}
\section{Some additional results}\label{additional results}

The purpose of this section is to present and show some new technical results used throughout the work.
\begin{proposition}\label{prop-app-2}
The following identity is valid 
\begin{multline*}
 	\int_0^\infty 
    x^{a-1}
    \exp(-sx)\gamma(b,\theta x)\gamma(c,\xi x){\rm d}x
    \\[0,2cm]
 	=
 	{\theta^b \xi^c \Gamma(a+b+c)\over bc (s+\theta+\xi)^{a+b+c}} \,
 	F_2\left(a+b+c,1,1;b+1,c+1;{\theta\over s+\theta+\xi},{\xi\over s+\theta+\xi}\right),
\end{multline*}
where $a>0$, $s>0$, $\theta>0$, $\xi>0$, $b>0$ and $c>0$.
\end{proposition}
\begin{proof}
By using the definition 
of lower incomplete gamma function, we have
\begin{align*}
    \gamma(b,\theta x)
    =
    \int_{0}^{\theta x}y^{s-1} \exp(-y){\rm d} y
    =
    \theta^b x^b 
    \int_{0}^{1} w^{b-1} \exp(-\theta x w) {\rm d} w,
\end{align*}
where the change of variable $w=y/(\theta x)$ was considered. Analogously,
\begin{align*}
    \gamma(c,\xi x)
    =
    \xi^c x^c
    \int_{0}^{1} u^{c-1} \exp(-\xi x u) {\rm d} u.
\end{align*}

Hence,
\begin{align*}
    &\int_0^\infty 
    x^{a-1}
    \exp(-sx)\gamma(b,\theta x)\gamma(c,\xi x){\rm d}x
    \\[0,2cm]
    &=
    \theta^b\xi^c
    \int_0^\infty
    \int_0^1
    \int_0^1
    [x(s+\theta w+\xi u)]^{a+b+c-1}
    \exp[-x(s+\theta w+\xi u)] \,
    {w^{b-1} u^{c-1}\over (s+\theta w+\xi u)^{a+b+c-1}}
    {\rm d}w
    {\rm d}u
    {\rm d}x.
\end{align*}
Making the change of variable $z=x(s+\theta w+\xi u)$ and changing the order of integration, the last integral becomes 
\begin{align*}
    &=
    \theta^b\xi^c
    \int_0^1
    \int_0^1
    \left[
    \int_0^\infty
    z^{a+b+c-1}
    \exp(-z) 
    {\rm d}z
    \right]
    {w^{b-1} u^{c-1}\over (s+\theta w+\xi u)^{a+b+c}}
    {\rm d}w
    {\rm d}u
        \\[0,2cm]
    &=
    \theta^b\xi^c\Gamma(a+b+c)
    \int_0^1
    \int_0^1
    {w^{b-1} u^{c-1}\over (s+\theta w+\xi u)^{a+b+c}}
    {\rm d}w
    {\rm d}u.
\end{align*}
By changing $w$ to $1-w$ and $u$ to $1-u$, the above integral can be written as
\begin{align*}
=
    {\theta^b\xi^c\Gamma(a+b+c)\over (s+\theta+\xi)^{a+b+c}}
    \int_0^1
    \int_0^1
    {(1-w)^{b-1} (1-u)^{c-1}\over \big(1-{\theta\over s+\theta+\xi} w-{\xi\over s+\theta+\xi} u\big)^{a+b+c}}
    {\rm d}w
    {\rm d}u.
\end{align*}

In short, we have proven that
\begin{align}\label{main-identity}
\int_0^\infty 
    x^{a-1}
    \exp(-sx)\gamma(b,\theta x)\gamma(c,\xi x){\rm d}x
    =
        {\theta^b\xi^c\Gamma(a+b+c)\over (s+\theta+\xi)^{a+b+c}}
    \int_0^1
    \int_0^1
    {(1-w)^{b-1} (1-u)^{c-1}\over \big(1-{\theta\over s+\theta+\xi} w-{\xi\over s+\theta+\xi} u\big)^{a+b+c}}
    {\rm d}w
    {\rm d}u.
\end{align}
By using in \eqref{main-identity} the well-known integral representation of the Appell F2 double hypergeometric function:
	\begin{multline*}
		F_2(a,b_1,b_2;c_1,c_2;x,y)
        \\[0,2cm]
		=
        {\Gamma(c_1)\Gamma(c_2)\over\Gamma(b_1)\Gamma(b_2)\Gamma(c_1-b_1)\Gamma(c_2-b_2)}
        \int_0^1
        \int_0^1
	u^{b_1-1}v^{b_2-1}(1-u)^{c_1-b_1-1}(1-v)^{c_2-b_2-1}(1-ux-vy)^{-a}{\rm d}u{\rm d}v,
	\end{multline*} 
     the proof follows.
\end{proof}

\begin{proposition}\label{prop-app-3}
The following identity is valid 
\begin{align*}
 	\int_0^\infty 
    x^{a-1}
    \exp(-sx)\gamma(b,\theta x){\rm d}x
    =
    {\theta^b\Gamma(a+b)\over b(s+\theta)^{a+b}}
\,_{2}F_{1}\left(a+b,1;b+1;{\theta\over s+\theta}\right),
\end{align*}
where $a>0$, $s>0$, $\theta>0$ and $b>0$.
\end{proposition}
\begin{proof}
By taking $c=1$ and then $\xi\to\infty$ in \eqref{main-identity}, we get
    \begin{align}\label{main-identity-3}
&\int_0^\infty 
    x^{a-1}
    \exp(-sx)\gamma(b,\theta x){\rm d}x
    =
        \lim_{\xi\to\infty}
        {\theta^b\xi\Gamma(a+b+1)\over (s+\theta+\xi)^{a+b+1}}
    \int_0^1
    \int_0^1
    {(1-w)^{b-1} \over \big(1-{\theta\over s+\theta+\xi} w-{\xi\over s+\theta+\xi} u\big)^{a+b+1}}
    {\rm d}w
    {\rm d}u
    \nonumber
        \\[0,2cm]
&=
 \lim_{\xi\to\infty}
    \theta^b\Gamma(a+b)
      \int_0^1
    (1-w)^{b-1} 
    \left\{
(s+\theta-\theta w)^{-a-b}
-
(s+\theta+\xi-\theta w)^{-a-b}
    \right\}
    {\rm d}w
\nonumber
                \\[0,2cm]
&=
        {\theta^b\Gamma(a+b)\over (s+\theta)^{a+b}}
      \int_0^1
    (1-w)^{b-1} 
\left(1-{\theta\over s+\theta} w\right)^{-a-b}
    {\rm d}w.
\end{align}
By using in \eqref{main-identity-3} the well-known integral representation of the  hypergeometric function:
\begin{align*}
    {\displaystyle \mathrm {B} (b,c-b)\,_{2}F_{1}(a,b;c;z)
    =
    \int_{0}^{1}x^{b-1}(1-x)^{c-b-1}(1-zx)^{-a}{\rm d}x,\quad c>b>0,}
\end{align*}
the proof readily follows.
\end{proof}

{\color{black}
\section{Benchmark models}\label{ap:competitor_models}

For the reader's convenience, we present below the CDFs of the competitors' copulas and the PDFs of the unit competing models. 

\subsection{Bivariate copula models}

The CDF of the Generalized Farlie-Gumbel-Morgenstern copula is given by:
$$C^{GFGM}(u,v; \rho) = uv + \rho \phi(u)\phi(v), ~~\rho\in[-1,1],$$
where $u$ and $v$ are the marginal distribution functions of random variables $X$ and $Y$, respectively. Observe that if $\rho=0$, then  $C^{GFGM}(u,v;0) = uv$, that is, $X$ and $Y$ are independent.



To get the PDF associated with each CDF and use it in the estimation procedures, it suffices to consider that 
\begin{equation*}
    c(u,v;\rho)=\frac{\partial^2C(u,v; \rho)}{\partial u \partial v}.
\end{equation*}

\subsection{Unit models}
For simplicity of notation, we set the parameters as $\alpha$, $\beta$, and $\rho$ for all competing distributions. The ML estimates for the models below were obtained using the \textsf{goodness.fit} function from the \textsf{AdequacyModel} package in R software.

The Beta and Kumaraswamy models have PDFs, respectively, given by
\begin{equation*}\label{eq:beta_pdf}
    f_B(x; \alpha, \beta) = \frac{\Gamma(\alpha+\beta)}{\Gamma(\alpha)\Gamma(\beta)} x^{\alpha-1} (1-x)^{\beta-1}, ~~0<x<1,~\alpha>0, \beta>0,
\end{equation*}
and
$$f_K(x; \alpha, \beta) = \alpha  \beta x^{\alpha -1 } (1-x^\alpha)^{\beta-1}, ~~0<x<1,~\alpha>0, \beta>0.$$

The one-parameter competing models used in Section \ref{sec:applications} are the Unit-Lindley and the Topp-Leone, which have PDFs given, respectively, by:


\[
f_{UL}(x; \alpha)=
\frac{\alpha^2}{1+\alpha}
(1-x)^{-3}
\exp\left(
-\frac{\alpha x}{1-x}
\right),~~0<x<1,~~\alpha>0
\]


\[
f_{TL}(x)=
2\alpha
x^{\alpha-1}
(1-x)
(2-x)^{\alpha-1},~~0<x<1,~~\alpha>0.
\]

The other two-parameter competing models used in Section \ref{sec:applications} are the Logit-Normal, Unit-Gompertz, Unit-inverse Gaussian, and Unit-Weibull, which have PDFs given, respectively, by:


\[
f_{LN}(x;\alpha,\beta)=
\frac{1}{\beta\sqrt{2\pi}}
\exp\left[
-\frac{(\operatorname{logit}(x)-\alpha)^2}
{2\beta^2}
\right]
\frac{1}{x(1-x)},~~0<x<1,~\alpha\in\mathbb{R}, \beta>0,
\]






\[
f_{UG}(x;\alpha,\beta)=
\alpha\beta
x^{-(\beta+1)}
\exp\left[
-\alpha\left(x^{-\beta}-1\right)
\right],~~0<x<1,~\alpha>0, \beta>0,
\]



\[
f_{UIG}(x;\alpha,\beta)=
\sqrt{
\frac{\beta}{2\pi}
}
\,
(x(1-x))^{-3/2}
\exp\left(
-\frac{
\beta(x-\alpha(1-x))^2
}{
2\alpha^2x(1-x)
}
\right),~~0<x<1,~\alpha>0, \beta>0,
\]

$$f_{UW}(x;\alpha,\beta) = \frac{1}{x} \alpha \beta (-\log x)^{\beta-1}\exp\left[-\alpha(-\log x)^\beta\right],~~0<x<1,~\alpha>0, \beta>0.$$

The three-parameter competing models used in Section \ref{sec:applications} are the Unit-Fréchet and Unit-Bimodal Birnbaum-Saunders, which have PDFs given, respectively, by:
	\begin{align*}\label{pdf-main}
		f_{UF}(x;\boldsymbol{\theta})
		=
		{\alpha\over \beta^\alpha}\,
		s^{\alpha-1}(s+1)^2
		\left\{ \displaystyle
		{\displaystyle2\left({s^\alpha\over \beta^\alpha}+1\right)^2
			-
			\rho\left[\left({s^\alpha\over \beta^\alpha}\right)^2+1\right]
			\over 
			\displaystyle
			\left[
			\left({s^\alpha\over \beta^\alpha}+1\right)^2
			-\rho\, {s^\alpha\over\beta^\alpha}
			\right]^2}
		- 
		{\displaystyle 1\over\displaystyle \left({s^\alpha\over \beta^\alpha} + 1\right)^2}
		\right\},
		\quad 
		s={x\over 1-x}, 
	\end{align*}
    where $\alpha>0, \beta>0$, and $0\leq\rho\leq1$, and
    $$f_{UBBS}(x; \alpha, \beta, \rho) = \frac{1}{4x\alpha\beta\Phi(x)}  \left[ \left(-\frac{\beta}{\log x}\right)^{1/2} + \left(-\frac{\beta}{\log x}\right)^{3/2} \right]\phi(|t(x)|+\rho),$$
where $x\in(0,1), t(x)=\frac{1}{\alpha}\left(\sqrt{-\frac{\log x}{\beta}} - \sqrt{-\frac{\beta}{\log x}}\right)$, $\alpha>0, \beta>0, \rho\in\mathbb{R}$.

The four-parameter competing model used in Section \ref{sec:applications} is the Bimodal Beta distribution, which has PDFs given by:
\[
f_{BB}(x;\alpha, \beta, \rho, \delta)=\dfrac{\rho+(1-\delta x)^2}
{Z\,B(\alpha,\beta)}
\,x^{\alpha-1}(1-x)^{\beta-1},~~ 0<x<1, 
\]
where $\alpha>0, \beta>0, \rho\geq0, \delta\in\mathbb{R}$,
\[
Z=Z(\alpha, \beta, \rho, \delta)
=
1+\rho
-2\delta\frac{\alpha}{\alpha+\beta}
+\delta^2
\frac{\alpha(\alpha+1)}
{(\alpha+\beta)(\alpha+\beta+1)}
\]
denotes the normalization constant and $B(\alpha,\beta)
=
\int_0^1
t^{\alpha-1}(1-t)^{\beta-1}\,dt$
is the beta function.
}

\end{appendices}
\end{document}